\documentclass[bj,numbers,noshowframe]{imsart}

\usepackage{amsmath, amsfonts, amsthm, verbatim,
	amssymb,dsfont, accents}

\usepackage[utf8x]{inputenc}
\usepackage{mathtools}

\startlocaldefs

\newcommand{\R}{\mathbb{R}}

\newcommand{\E}{\mathbb{E}}
\newcommand{\F}{\mathbb{F}}

\renewcommand{\P}{\mathbb{P}}
\newcommand{\Q}{\mathbb{Q}}

\newcommand{\bone}{\mathbf 1}

\theoremstyle{plain}
\newtheorem{theorem}{Theorem}[section]
\newtheorem{lemma}[theorem]{Lemma}
\newtheorem{proposition}[theorem]{Proposition}
\newtheorem{corollary}[theorem]{Corollary}
\newtheorem{condition}{Condition}
\newtheorem{Assumption}{Assumption}
\renewcommand{\theAssumption}{\Alph{Assumption}}
\makeatletter
\newcommand{\settheoremtag}[1]{
	\let\oldtheAssumption\theAssumption
	\renewcommand{\theAssumption}{#1}
	\g@addto@macro\endAssumption{
		\addtocounter{Assumption}{-1}
		\global\let\theAssumption\oldtheAssumption}
}
\makeatother
\theoremstyle{remark}
\newtheorem{remark}[theorem]{Remark}

\newcommand{\calh}{{\cal H}}

\newcommand{\calf}{{\cal F}}

\newcommand{\calp}{{\cal P}}
\newcommand{\calb}{{\cal B}}
\newcommand{\call}{{\cal L}}

\newcommand{\calg}{{\cal G}}

\newcommand{\al}{{\alpha}}
\newcommand{\la}{{\lambda}}

\newcommand{\ga}{{\gamma}}
\newcommand{\Ga}{{\Gamma}}
\newcommand{\vp}{{\varphi}}
\newcommand{\si}{{\sigma}}

\newcommand{\Om}{{\Omega}}

\newcommand{\ov}{\overline}
\newcommand{\un}{\underline}
\newcommand{\wh}{\widehat}
\newcommand{\wt}{\widetilde}
\newcommand{\wa}{\accentset{\ast}}

\newcommand{\lec}{\lesssim}

\newcommand{\Leb}{\mathrm{Leb}}

\newcommand{\bthm}{\begin{theorem}}
	\newcommand{\ethm}{\end{theorem}}

\newcommand{\bcor}{\begin{corollary}}
	\newcommand{\ecor}{\end{corollary}}

\newcommand{\blem}{\begin{lemma}}
	\newcommand{\elem}{\end{lemma}}

\newcommand{\bprop}{\begin{proposition}}
	\newcommand{\eprop}{\end{proposition}}

\newcommand{\bcond}{\begin{condition}}
	\newcommand{\econd}{\end{condition}}

\newcommand{\bdf}{\begin{definition}}
	\newcommand{\edf}{\end{definition}}

\newcommand{\bex}{\begin{example}}
	\newcommand{\eex}{\end{example}}

\newcommand{\brem}{\begin{remark}}
	\newcommand{\erem}{\end{remark}}

\newcommand{\bpr}{\begin{proof}}
	\newcommand{\epr}{\end{proof}}

\newcommand{\benu}{\begin{enumerate}}
	\newcommand{\eenu}{\end{enumerate}}

\newcommand{\beq}{\begin{equation}}
	\newcommand{\eeq}{\end{equation}}

\newcommand{\bit}{\begin{itemize}}
	\newcommand{\eit}{\end{itemize}}

\newcommand{\bass}{\begin{Assumption}}
	\newcommand{\eass}{\end{Assumption}}

\numberwithin{equation}{section}

\DeclareMathOperator\Arg{Arg}

\allowdisplaybreaks[3]
\endlocaldefs

\begin{document}

	\begin{frontmatter}
	\title{Short-time expansion of characteristic functions in a rough volatility setting with applications}
	\runtitle{Short-time expansion of characteristic functions}
	
	\begin{aug}
		\author[A]{\fnms{Carsten H.} \snm{Chong}\ead[label=e1]{carstenchong@ust.hk}}
		\author[B]{\fnms{Viktor} \snm{Todorov}\ead[label=e2]{v-todorov@kellogg.northwestern.edu}}
		
		\address[A]{Department of Information Systems, Business Statistics and Operations Management,
			The Hong Kong University of Science and Technology, Hong Kong,
			\printead{e1}}
		
		\address[B]{Department of Finance, Northwestern University, USA,
			\printead{e2}}
	\end{aug}

	\begin{abstract}
		We derive a higher-order asymptotic expansion of the conditional characteristic function of the increment of an It\^o semimartingale over a shrinking time interval. The spot characteristics of the It\^o semimartingale are allowed to have dynamics of general form. In particular, their paths can be rough, that is, exhibit local behavior like that of a fractional Brownian motion, while at the same time have jumps with arbitrary degree of activity. The expansion result shows the distinct roles played by the different features of the spot characteristics dynamics. As an application of our result, we construct a  nonparametric estimator of the Hurst parameter of the diffusive volatility process from portfolios of short-dated options written on an underlying asset.
	\end{abstract}

	
\end{frontmatter}

\section{Introduction}\label{sec:intro}

Our interest in this paper is a small-time asymptotic expansion of the characteristic function of the increment of an It\^o semimartingale, defined on some filtered probability space $(\Om,\calf,\F=(\calf_t)_{t\geq0},\Q)$ satisfying the usual conditions and given by 
\beq\label{eq:x} x_t = x_0 + \int_0^t \al_s ds+ \int_0^t \si_s dW_s+\int_0^t\int_\R \ga(s,z)\mu(ds,dz) + \int_0^t\int_\R \delta(s,z)(\mu-\nu)(ds,dz), \eeq
where $W$ is a Brownian motion, $\mu$ is a Poisson measure on $\mathbb{R}_+\times\mathbb{R}$ with compensator $\nu$, $\al$ and $\si$ are some processes with c\`{a}dl\`{a}g paths, and $\gamma, \delta: \mathbb{R}_+\times\mathbb{R}\rightarrow\mathbb{R}$ are predictable functions. The technical assumptions for all quantities appearing in (\ref{eq:x}) are given in Section~\ref{sec:model} below. 

Under fairly weak assumptions on $\alpha$, $\sigma$, $\ga$ and $\delta$, one can show  for $t\geq0$ and $u\in\R\setminus\{0\}$ that
\begin{equation}\label{eq:fo_exp}
	\mathbb{E}_t\Bigl[e^{iu(x_{t+T}-x_t)/\sqrt{T}}\Bigr] = e^{-\frac{u^2}{2}\sigma_t^2}+o_p(1) \qquad\text{as $T\downarrow 0$},
\end{equation}
where  $\mathbb{E}_t$ denotes the $\mathcal{F}_t$-conditional expectation with respect to $\Q$. 
This result is due to the fact that the variation in the spot semimartingale characteristics $\alpha$, $\sigma$, $\ga$ and $\delta$ is of higher order when the length of the increment shrinks (see e.g., \cite{JP12}) and due to the leading role   the diffusion coefficient plays  in the characteristic function of a L\'{e}vy process as the value of the characteristic exponent increases (see e.g., Theorem 8.1 in \cite{SATO}). The result in \eqref{eq:fo_exp} has been used for efficient volatility estimation from high-frequency observations of $x$ by \cite{jacod2014efficient,jacod2018limit} and from short-dated options when $x$ is an asset price, by \cite{T19}.    

The goal of the current paper is to investigate how the higher-order terms in \eqref{eq:fo_exp} depend on the dynamics of the spot semimartingale characteristics. This allows one to   quantify their impact on volatility estimation, in which the $o_p(1)$-term in (\ref{eq:fo_exp}) is ignored. But more importantly, it allows one to obtain information about the spot characteristics themselves and their dynamics. Such a higher-order asymptotic expansion of the characteristic function has been derived in \cite{T21} in the case where $x$ and its spot characteristics are governed by It\^o semimartingales with  finite variation jumps. 
In this paper, we consider an asymptotic expansion in a much more general setting, in which not only the jumps can be of infinite variation but also the volatility $\sigma$ and the jump size functions $\gamma$ and $\delta$ can have local behavior like that of a fractional Brownian motion. Models with the latter feature can generate   rough volatility paths, with H\"older exponents (outside jump times) that are significantly smaller than $\frac12$. Rough volatility models have been introduced recently by \cite{gatheral2018volatility} and have been found to generate patterns like those observed in real financial data (see e.g., \cite{Bennedsen21, euch2019short,el2018microstructural,el2019roughening, FTW21}  among others). Models with rough dynamics have also been  used by \cite{CDL22} to capture market microstructure noise in financial data with plausible dependence structure. 

The setup we work in allows for rough dynamics of general form. In particular, all spot characteristics can be rough and their degree of roughness may differ. This generality allows us to determine whether the various model features have similar or different contributions to the asymptotic expansion. In particular, our higher-order asymptotic expansion result shows that the contribution of the jumps in $x$ and the roughness of the volatility path to the characteristic function can be of the same asymptotic order, depending on the degree of jump activity and roughness of the volatility path. That is, for each value of volatility roughness, there is a degree of jump activity such that their contributions to the characteristic function will be of the same asymptotic order. However, our analysis also shows that the terms  due to  jumps in $x$ and those due to  roughness of volatility differ in how they depend on the characteristic function exponent, which in the end allows one to separate them. Compared to results in the constrained setting where the spot semimartingale characteristics are It\^o semimartingales without rough components in their dynamics, higher-order terms play a much more prominent role in the presence of roughness. In other words, in rough settings the $o_p(1)$-term in \eqref{eq:fo_exp} is much bigger, and hence ignoring it would lead to a much bigger error.

The derived asymptotic expansion can be used for developing nonparametric estimators of quantities related to volatility and its dynamics as well as to jumps in rough settings using high-frequency observations of $x$ or short-dated options when $x$ is the price of a financial asset. We illustrate this by deriving a consistent estimator of the   roughness parameter of volatility, using options written on an asset with short times to maturity. Following \cite{QT19} and \cite{T19}, we use a portfolio of options with different strikes and with the same expiration to estimate $\mathbb{E}_t[e^{iu(x_{t+T}-x_t)/\sqrt{T}}]$. We then infer the degree of  roughness of   volatility from the rate of decay of the argument of the characteristic function as the length of the price increment shrinks down to zero. This estimation approach works well when the underlying asset price does not contain jumps. If this is not the case, then one first needs to perform a bias correction of the argument of the conditional characteristic function by suitably differencing the latter computed over different short time horizons. While this differencing leads to loss of information, compared to the estimation with no price jumps, it shows that separation of the roughness parameter of volatility from  price jumps is possible theoretically. 

Our use of short-dated options  in a rough setting can be compared with recent work on short-time expansions of option prices in various rough volatility specifications (see for example, \cite{BFGHS19, euch2019short, FGS21, FSV21, FZ17, FGP21, FGP22, F21, LMPR18}). This strand of work considers expansions of a single option, and typically these expansions  only hold for a restricted class of models (e.g., the volatility only contains  a rough component without jumps, and in addition there are no jumps in the asset price). By contrast, our model setting, as already explained, is very general and we consider a short-term asymptotic expansion for a portfolio of options that mimics the conditional characteristic function of the price increment.

The rest of the paper is organized as follows. We start in Section~\ref{sec:model} with introducing our setting and stating the assumptions. The main result is given in Section~\ref{sec:expand}. In Section~\ref{sec:options}, we introduce option-based estimators of the conditional characteristic function of the underlying asset price and derive their rate of convergence. Section~\ref{sec:application} contains our application to estimation of the degree of volatility roughness from options. In Section~\ref{sec:mc} we conduct several numerical experiments to assess the precision of the higher-order expansions. Proofs are given in Sections~\ref{sec:proof:cf}--\ref{sec:9}.

\section{Setup and assumptions}\label{sec:model} 

We start by introducing the setup and stating the assumptions needed for deriving our asymptotic expansion result presented in the next section. The dynamics of the volatility process $\si$ are given by a function of an It\^o semimartingale plus a rough component:
\beq\label{eq:sigma} \begin{cases}  \si_t=f(v_t),\quad\text{where $f\in C^2$ and} \\
	\begin{aligned} v_t&=v_0+\int_0^t b_s ds +\int_0^t g_H(t-s)\si^v_s dW_s + \int_0^t \wt g_H(t-s)\wt\si^v_s d\wt W_s+ \int_0^t \eta^v_s dW_s \\
		&\quad+ \int_0^t \wt\eta^v_s d\wt W_s  + \int_0^t \ov \eta^v_s d\ov W_s+ \int_0^t\int_\R \ga_v(s,z) \mu(ds,dz) + \int_0^t\int_\R \delta_v(s,z)  (\mu-\nu)(ds,dz).\end{aligned}\!\!\!\!\!\!\!\!\!\!\!\!\!\!\!\! \end{cases}\eeq
Except for the two integrals $\int_0^t g_H(t-s)\si^v_s dW_s$ and $\int_0^t \wt g_H(t-s)\wt\si^v_s d\wt W_s$, the dynamics of $v$ in (\ref{eq:sigma}) is that of a general It\^o semimartingale. The integrals $\int_0^t g_H(t-s)\si^v_s dW_s$ and $\int_0^t \wt g_H(t-s)\wt\si^v_s d\wt W_s$, upon a suitable choice of the kernel functions $g_H$ and $\wt g_H$, can exhibit local behavior like that of a fractional Brownian motion and therefore introduce roughness in the paths of $\si$. We note that in the above specification of $\si$ we allow for jumps, and as will become clear later on, they can be of infinite variation. We also allow for dependence of general form between the diffusive and jump components of $x$ and $v$. Consequently, our model for $v$ is an extension of the class of  mixed semimartingales from \cite{CDM22} that includes jumps. The specification \eqref{eq:sigma} contains many parametric rough volatility models as particular examples, such as the rough fractional stochastic volatility model \cite{gatheral2018volatility}, the rough Heston model \cite{el2019roughening} or the fractional Ornstein--Uhlenbeck process \cite{WXY23}.

In our asymptotic expansion, the dynamics of the processes $\si^v$ and $\wt\si^v$ will also play a role. We make a similar assumption for their dynamics to the one for $\si$ above. More specifically, we have
\beq\label{eq:eta}\begin{cases}
	\si^v_t=h(\eta_t),\quad\text{where $h\in C^2$ and} \\
	\begin{aligned}\eta_t&=\eta_0 +\int_0^t g_H^\eta(t-s)\si^\eta_s dW_s + \int_0^t \wt g_H^\eta(t-s)\wt\si^\eta_s d\wt W_s+\int_0^t \ov g_H^\eta(t-s)\ov \si_s^\eta d\ov W_s \\
		&\quad +\int_0^t \wh g^\eta_H(t-s)\wh \si^\eta_s d\wh W_s+  \theta_t,\end{aligned}
\end{cases}\eeq
and
\beq\label{eq:etatilde}\begin{cases}
	\wt\si^v_t=\wt h(\wt\eta_t),\quad\text{where $\wt h\in C^2$ and}\\
	\begin{aligned}\wt\eta_t&=\wt\eta_0 +\int_0^t g_H^{\wt \eta}(t-s)\si^{\wt \eta}_s dW_s + \int_0^t \wt g_H^{\wt \eta}(t-s)\wt\si^{\wt \eta}_s d\wt W_s+\int_0^t \ov g_H^{\wt \eta}(t-s)\ov \si_s^{\wt \eta} d\ov W_s \\
		&\quad +\int_0^t \wh g^{\wt \eta}_H(t-s)\wh \si^{\wt \eta}_s d\wh W_s +\int_0^t \mathring g^{\wt \eta}_H(t-s)\mathring \si^{\wt \eta}_s d\mathring W_s+ \wt\theta_t.\end{aligned}
\end{cases}\eeq
Finally, the dynamics of the ``small'' jump size function $\delta$ will also play a role in our expansion. We assume that
\beq\label{eq:delta}\begin{split}
	\delta(t,z)&=\delta(0,z)+\int_0^t g^\delta_{H_\delta}(t-s,z)\si^\delta(s,z)dW_s+\int_0^t \dot g^\delta_{H_\delta}(t-s,z)\dot\si^\delta(s,z)d\dot W_s\\
	&\quad+\int_0^t \eta^\delta(s,z)dW_s+\int_0^t \ddot\eta^\delta(s,z)d\ddot W_s+\int_0^{t-}\int_\R \delta_\delta(s,z,z')(\mu-\nu)(ds,dz') +\theta_\delta(t,z). 
\end{split}\raisetag{2\baselineskip}\eeq
The specification in \eqref{eq:delta} is very general, allowing for both diffusive and jump shocks to determine changes over time in $\delta$. We also note that $\delta$ can exhibit roughness in its paths through the first two integrals.

The remaining ingredients of \eqref{eq:x} and \eqref{eq:sigma}--\eqref{eq:delta} are supposed to satisfy the following conditions, which depend on two roughness parameters $H,H_\delta\in(0,\frac12)$:

\bass\label{ass:A} We have \eqref{eq:x} and \eqref{eq:sigma}--\eqref{eq:delta}, where
\benu
\item $x_0$, $v_0$, $\eta_0$ and $\wt\eta_0$ are $\calf_0$-measurable  and $\delta(0,\cdot)$ is $\calf_0\otimes\calb(\R)$-measurable; 
\item $\al$, $b$,  $\si^\eta$, $\wt\si^\eta$, $\ov\si^\eta$, $\wh \si^\eta$, $\si^{\wt\eta}$, $\wt\si^{\wt\eta}$, $\ov\si^{\wt\eta}$, $\wh\si^{\wt\eta}$, $\mathring{\si}^{\wt\eta}$, $\si^\delta$, $\dot\si^\delta$, $\eta^v$, $\wt \eta^v$, $\ov\eta^v$, $\eta^\delta$, $\ddot \eta^\delta$, $\theta$, $\wt \theta$, $\theta_\delta$, $\ga$, $\ga_v$, $\delta_v$  and $\delta_\delta$ are predictable processes;
\item $W$, $\wt W$, $\ov W$, $\wh W$ and $\mathring W$ are independent standard $\F$-Brownian motions; $\dot W$ and $\ddot W$ are  standard $\F$-Brownian motions that are jointly Gaussian with $(W, \wt W, \ov W, \wh W, \mathring W)$, independent of $W$ but potentially dependent on each other and   on $(\wt W, \ov W, \wh W, \mathring W)$;
\item $\mu$ is an $\F$-Poisson random measure on $[0,\infty)\times\R$ with compensator $\Leb\otimes \nu$ that is independent of $W$, $\wt W$, $\ov W$, $\wh W$, $\mathring W$, $\dot W$ and $\ddot W$;  
\item  each $g\in\{g_H, \wt g_H, g^\eta_H, \wt g^\eta_H, \ov g^\eta_H, \wh g^\eta_H, g^{\wt\eta}_H, \wt g^{\wt\eta}_H, \ov g^{\wt\eta}_H,\wh g^{\wt\eta}_H,\mathring g^{\wt\eta}_H\}$ is of the form $g=k_H + \ell_g$ and each $g\in\{g_{H_\delta}^\delta(\cdot,z),\dot g^\delta_{H_\delta}(\cdot,z):z\in\R\}$ is of the form $g=k_{H_\delta} + \ell_g$, where
\beq\label{eq:kH}  k_H(t)=\frac1{\Ga(H+\frac12)} t^{H-1/2}\bone_{\{t>0\}},\qquad k_{H_\delta}(t)=\frac1{\Ga(H_\delta+\frac12)} t^{H_\delta-1/2}\bone_{\{t>0\}} \eeq
and $\ell_g$ is continuous on $[0,\infty)$ with $\ell_g(0)=0$ and  differentiable on $(0,\infty)$ such that $\lvert\ell'_g(s)\rvert\leq Ls^{H-1/2}$ for all $s\in[0,t]$ and some  $L\in(0,\infty)$ that only depends on $t$ but not on $g$.
\eenu
\eass

Assumption \ref{ass:A} is a regularity condition with the exception of part 5  in which the rough components in the dynamics of the various processes are specified. The semiparametric specification of the function $g$ above covers all kernels of the form $g(t)=\Ga(H+\frac12)^{-1}t^{H-1/2}h_g(t)$, where $h_g\in C^1$ with $h_g(0)=1$ and a locally bounded derivative (take $\ell_g (t)=g(t)-k_H(t)=\Ga(H+\frac12)^{-1}t^{H-1/2}(h_g(t)-1)$). This includes, in particular, the power law kernels and gamma kernels considered in \cite{Bennedsen21}. 
We note that we allow different degrees of roughness for the processes driving the volatility process $\si$ and the small jump size function $\delta$, given by $H$ and $H_\delta$, respectively. 

We will further impose the following moment and smoothness assumptions, which are parametrized by $H$ and $H_\delta$ from Assumption~\ref{ass:A} and three additional parameters: $q\in(0,1]$, $r\in[1,2)$ and $H_\ga\in(0,\frac12)$. To simplify the notation, we write $o=o_p$ and $O=O_p$ in the following.
\bass\label{ass:B} For all $t,T>0$, there exists a positive finite $\calf_t$-measurable random variable $C_t$ and a positive $\calf_t\otimes\calb(\R)$-measurable function $C_t(z)$ such that:
\benu
\item For   $s\in[t,t+T]$ and $v\in\{b, \eta^v, \wt \eta^v, \ov\eta^v,\si^\eta,\wt\si^\eta,\ov\si^\eta,\wh\si^\eta,\si^{\wt \eta},\wt\si^{\wt \eta},\ov\si^{\wt \eta},\wh\si^{\wt \eta},\mathring\si^{\wt \eta},\theta,\wt\theta\}$, we have 
\beq\label{eq:mom1}\E_t[\lvert\al_s\rvert]+\E_t[v_s^2]^{1/2}<C_t.\eeq
\item For any   $v\in\{ \si^\eta, \wt\si^\eta, \ov\si^\eta, \wh\si^\eta, \si^{\wt\eta}, \wt\si^{\wt\eta}, \ov\si^{\wt\eta}, \wh\si^{\wt\eta}, \mathring\si^{\wt\eta}\}$, we have  
\beq\label{eq:smooth1} \sup_{s,s'\in[t,t+T]} \E_t[(v_s-v_{s'})^2]^{1/2} = o(1) \eeq
as $T\to0$. Moreover, 
with the same $H$ as in Assumption~\ref{ass:A}, we have
\beq\label{eq:mom2}\sup_{s,s'\in[t,t+T]}  \Bigl\{ \E_t[(\theta_s-\theta_{s'})^2]^{1/2} +\E_t[(\wt\theta_s-\wt\theta_{s'})^2]^{1/2}\Bigr\} =o(T^H)\eeq
and
\beq\label{eq:smooth2}\begin{split}
	&\sup_{s,s'\in[t,t+T]}  \biggl\{\E_t[\lvert \al_s-\al_{s'}\rvert] + \E_t[\lvert \eta^v_s-\eta^v_{s'}\rvert^2]^{1/2}+\E_t[\lvert \wt\eta^v_s-\wt\eta^v_{s'}\rvert^2]^{1/2} \\
	&\qquad + \E_t[\lvert \ov\eta^v_s-\ov\eta^v_{s'}\rvert^2]^{1/2}+ \E_t\Biggl[ \int_\R  (\delta_v(s,z)-\delta_v (s',z))^2   \nu(dz)\Biggr]^{1/2}\biggr\} =O(T^H).\end{split}
\eeq
\item For all $s\in[t,t+T]$,
\beq\label{eq:mom3}\begin{split} &\lvert\delta(t,z)\rvert+\E_t [   \si^\delta(s,z)^2]^{1/2}+\E_t[ \dot\si^\delta(s,z)^2]^{1/2}+\E_t[\eta^\delta(s,z)^2]^{1/2}\\
	&\qquad+\E_t[ \ddot\eta^\delta(s,z)^2]^{1/2}+\E_t\biggl[\int_\R \lvert \delta_\delta(s,z,z')\rvert^r \nu(dz')\biggr]^{1/r}<C_t(z) \end{split} \eeq
and
\beq\label{eq:mom4} \E_t\Biggl[ \int_\R (\lvert \ga(s,z)\rvert^q+C_t(z)^r+\lvert\ga_v(s,z)\rvert^q + \lvert\delta_v(s,z)\rvert^2)  \nu(dz)\Biggr]<C_t.  \eeq
\item As $T\to0$,
\begin{equation}\label{eq:smooth3}\sup_{s,s'\in[t,t+T]}  \E_t\Biggl[  \int_\R   \lvert \ga(s,z)   -\ga(s',z)   \rvert^q   \nu(dz)\Biggr]^{1/q} = O(T^{H_\ga}).\end{equation}
Moreover, for all $s,s'\in[t,t+T]$,
\begin{equation}
	\label{eq:smooth4}  
	\begin{split}&\E_t[(\si^\delta(s,z)-\si^\delta(s',z))^2]^{1/2} +  \E_t[(\dot \si^\delta(s,z)-\dot \si^\delta(s',z))^2]^{1/2}\\
		&\qquad+\E_t[(\eta^\delta(s,z)-\eta^\delta(s',z))^2]^{1/2} +  \E_t[(\ddot \eta^\delta(s,z)-\ddot \eta^\delta(s',z))^2]^{1/2}\\
		&\qquad+\E_t\biggl[\int_\R \lvert \delta_\delta(s,z,z')-\delta_\delta(s',z,z')\rvert^r \nu(dz')\biggr]^{1/r}\leq C_t(z)T^{H_\delta} \end{split} \end{equation}
and 
\beq\label{eq:smooth5} \E_t[\lvert \theta_\delta(s,z)-\theta_\delta(s',z)\rvert^r]^{1/r}\leq C_t(z)T^{2H_\delta}.\eeq
\item The three roughness parameters $H$, $H_\ga$ and $H_\delta$ satisfy the relations 
\beq\label{eq:H} H_\ga > H-(\tfrac2q-1)(\tfrac12-H)\qquad\text{and}\qquad H_\delta > H-(\tfrac2r-1)[(\tfrac12-H)\wedge \tfrac14].\eeq
\eenu
\eass

Several comments are in order regarding Assumption \ref{ass:B}. Part 1 imposes finite moments of order up to $2$, conditionally on $\calf_t$, on all coefficients in the model. Part 2 requires some of these coefficients to be no rougher than volatility (in an $L^2$-sense). Part 3 and 4 introduce similar moment and regularity conditions on various jump-related coefficients, where the constants $q$ and $r$ are upper bounds on the degree of activity of the finite and infinite variation jump components, respectively, of $x$. The constant $H_\ga$ can be viewed  as a bound on the smoothness in time of the ``big'' jump size function $\ga$. Unlike the processes $\si$ and $\delta$, we do not model explicitly the dynamics of $\ga$ but only provide a bound for its smoothness. Both jumps in $\ga$ and rough components driven by Brownian motions, like the ones in $\si$ and $\delta$, will determine the value of $H_\ga$. The reason for treating  $\si$ and $\delta$ on one hand and $\ga$ on the other hand differently is that for the same degree of smoothness in expectation of these processes, the contribution in the higher-order expansion of the characteristic function due to $\ga$ will be dominated by that of $\si$ and $\delta$. Finally, part 5 in Assumption~\ref{ass:B} puts a lower bound on $H_\delta$ and $H_\ga$ (higher values of these parameters correspond to smoother paths), which are chosen in such a way that the variation of $\ga$ and certain components of the dynamics of $\delta$ can be neglected in the asymptotic expansion. This assumption is trivially satisfied in arguably the most typical case in which $H$, $H_\delta$ and $H_\ga$ are the same. Parametric rough volatility models typically impose   constancy or a deterministic relationship as a function of volatility on coefficients such as the volatility drift $b$, the leverage effect (i.e., $\si^v$) or volatility of volatility (i.e., $\si^v$ and $\wt\si^v$). Our nonparametric model outlined above, by contrast, does not make such assumptions and allows these coefficients to freely vary as stochastic processes. The same holds true for price and volatility jumps in our model, which are  absent in most parametric rough volatility models. For example, our assumptions permit jumps in price and volatility like  those of $\int_0^t K_{s} dL_s$, where $L$ is a Lévy process and $K$ is predictable.  

Overall, Assumptions~\ref{ass:A} and \ref{ass:B} allow $x$ to be an It\^o semimartingale with very general spot semimartingale characteristics. In particular, both the standard situation in which the spot semimartingale characteristics are themselves  It\^o semimartingales (see e.g., \cite{JP12}) and the pure rough setting in which $\si$ has no jumps and no martingale component (see e.g., \cite{gatheral2018volatility}) are nested in our model. 
We even allow  for models in which the jump size function can have roughness in its paths, which to the best of our knowledge has not been considered in prior parametric modeling.          

\section{Short-time expansion of characteristic functions}\label{sec:expand}
In this section we present our main theoretical result on a small-time asymptotic expansion of the conditional characteristic function of an increment of $x$, which we denote with 
\begin{equation}\label{eq:condchar} 
	\call_{t,T}(u)=\E_t[e^{iu(x_{t+T}-x_t)/\sqrt{T}}],\quad u\in\R.
\end{equation}
As before, the above conditional expectation is taken under $\Q$. While $\call_{t,T}(u)$ is typically not known in closed form, it can be computed analytically for some specific parametric models such as the rough Heston volatility model, see \cite{ER19}.

To state our expansion result, we need some notation. For $s\geq t\geq0$ and $u\in\R$, we set
\begin{equation}\label{eq:proj} 
	\si_{s\mid t}	= \E_t[\si_{s}],\quad 
	\si^v_{s\mid t}  = \E_t[\si^v_{s}],\quad v_{t'\mid t} = \E_t[v_{s}],\quad
	\delta(s\mid t,z)	=\E_t[\delta(s,z)]. 
\end{equation}
We note that as $T\rightarrow 0$,
\begin{equation}
\sigma_{t+T\mid t} = \begin{cases} \sigma_t + O(T^H) &\text{if } H<\frac12,\\[+1.5ex] \sigma_t + O(T) &\text{if } H=\frac12,\end{cases}
\end{equation}
with similar expansions for $\mathbb{E}_t[\si^v_{t+T}]$, $\mathbb{E}_t[v_{t+T}]$ and $\mathbb{E}_t[\delta(t+T,z)]$. In other words, 
if volatility is rough, the gap between $\sigma_{t+T\mid t}$ and $\sigma_t$ can be rather nontrivial, particularly for lower values of $H$. This is a major difference from a standard volatility model without a rough component. Moreover, at $H=\frac12$, there is a discontinuity in the rate of convergence of $\sigma_{t+T\mid t}$ towards $\sigma_t$.

We next introduce notation related to the characteristic exponents of the small and big jumps:
\beq\label{eq:phipsi} \begin{split} \phi_t(u)&=\int_\R e^{iu\delta(t,z)}(e^{iu\ga(t,z)}-1)\nu(dz), \quad\varphi_t(u)=\int_\R (e^{iu\delta(t,z)}-1-iu\delta (t,z))\nu(dz),\\
	\psi_t(u)&=\phi_t(u)+\varphi_t(u).
\end{split} \eeq  
Finally, in our expansions, small jumps in $x$ will play a key role, and we need the following additional notation for them:
\beq\label{eq:chi} \begin{split}  \chi^{(1)}_{t}(u)&=\int_\R \si^\delta(t,z)(e^{iu(\delta(t,z)+\ga(t,z))}-1)\nu(dz),\\
	\chi^{(1)}_{t'\mid t}(u)&=\int_\R  (\delta(t'\mid t,z)-\delta(t,z))(e^{iu(\delta(t,z)+\ga(t,z))}-1)\nu(dz),\\
	\chi^{(2)}_t(u)&=\int_\R \eta^\delta(t,z)(e^{iu(\delta(t,z)+\ga(t,z))}-1)\nu(dz),\\
	\chi^{(3)}_t(u)&=\int_\R\int_\R \delta_\delta(t,z,z')(e^{iu(\delta(t,z)+\ga(t,z))}-1)(e^{iu(\delta(t,z')+\ga(t,z'))}-1)\nu(dz)\nu(dz'),\\
	\chi^{(4)}_t(u)  &= \int_\R \delta_v(t,z)(e^{iu(\delta(t,z)+\ga(t,z))}-1)\nu(dz).
\end{split}\eeq

The following theorem contains our main expansion result for $\call_{t,T}(u)$. In the statement of the theorem, the superscript in $o^\mathrm{uc}$ stands for ``uniformly for $u$ belonging to compact subsets of $\R$.''
\bthm\label{thm:cf}
If Assumptions~\ref{ass:A} and \ref{ass:B} hold with $H, H_\ga,H_\delta\in(0,\frac12)$, $q\in(0,1]$ and $r\in[1,2)$, then as $T\to0$,
\beq\label{eq:cf}\begin{split}
	\call_{t,T}(u)&=\exp\biggl(iu\al_t\sqrt{T}-\frac12 u^2\int_0^1 \si^2_{t+sT\mid t}ds+T\psi_t(\tfrac {u}{\sqrt{T}}) \\
	&\quad\qquad -iu^3 \biggl(\frac12\si^2_tf'(v_t)\eta_t^v T^{1/2}+\frac{\si_t^2f'(v_t)\si^v_t}{\Ga(H+\frac{5}{2})}T^H+ C'_{1,0}(T)_tT^{2H} \biggr)  \biggr)\\
	&\quad 	+e^{-\frac12 u^2\si_t^2+T\varphi_t(u/\sqrt{T})} \biggl(\sum_{i=1}^4 C_{1,i}(u,T)_t+  \  C'_{1,1}(u,T)_t \biggr) +C_2(u)_tT^{2H} +o^{\mathrm{uc}}(T^{2H}),
\end{split}\raisetag{3.5\baselineskip}\eeq
where
\beq\label{eq:C1i}\begin{aligned}
	C_{1,1}(u,T)_t&=-\frac{u^2\si_tT^{H_\delta+1/2}\chi^{(1)}_t(\tfrac{u}{\sqrt{T}})}{\Ga(H_\delta+\frac52)},
	&C_{1,2}(u,T)_t&=-\frac12u^2\si_t   T\chi^{(2)}_t(\tfrac{u}{\sqrt{T}}),\\
	C_{1,3}(u,T)_t&=\frac12iu T^{3/2}\chi^{(3)}_t(\tfrac{u}{\sqrt{T}}), 
	&C_{1,4}(u,T)_t&=-\frac12 u^2\si_tf'(v_t)T\chi^{(4)}_t(\tfrac {u}{\sqrt{T}})
\end{aligned}
\eeq
and
\begin{align} \nonumber
	C'_{1,0}(T)_t&=  {T^{-H}}  \int_0^1 \int_0^sk_H(s-r)(f'(v_{t+sT\mid t})\si_{t+rT\mid t}^v\si_{t+sT\mid t} \si_{t+rT\mid t}   -f'(v_t)\si^v_t\si_t^2 )drds,\\
	C'_{1,1}(u,T)_t&=iuT^{1/2}\int_0^1  \chi^{(1)}_{t+sT\mid t}(\tfrac u{\sqrt{T}})ds \label{eq:C1'}
\end{align}
and
\beq\label{eq:C2} \begin{split}
	C_2(u)_t&=e^{-\frac12u^2\si_t^2}\biggl( \frac{\si_t^2f'(v_t)[h'(\eta_t)\si_t\si^\eta_t+f'(v_t)(\si^v_t)^2]u^4}{\Ga(2H+3)} -  \frac{f'(v_t)^2[(\si^v_t)^2+(\wt\si^v_t)^2]u^2}{8H\Ga(H+\frac12)\Ga(H+\frac32)}\\
	&\qquad\qquad\qquad\qquad\qquad\qquad\quad+\frac{\si_t^2 [f'(v_t)^2[(\si^v_t)^2+(\wt \si^v_t)^2]+f''(v_t)\si_t(\si^v_t)^2]u^4}{4(H+1)\Ga(H+\frac32)^2}  \biggr). 
\end{split}
\eeq
\ethm

\brem From the proof of Theorem~\ref{thm:cf}, it follows that all processes in \eqref{eq:x} and \eqref{eq:sigma}--\eqref{eq:delta} are well defined  under Assumptions~\ref{ass:A} and \ref{ass:B}.
\erem

\begin{remark} Related to Theorem~\ref{thm:cf}, short-term expansions of  near-the-money single options are derived in \cite{BFGHS19,euch2019short,FGP21, FGP22,FGS21,FZ17,LMPR18} for continuous  rough volatility models without jumps and in \cite{FSV21} for a semiparametric volatility model with a specific parametric  model for L\'{e}vy jumps in the price. The expansion derived  in \cite{FSV21} is only valid for a specific region of the Hurst parameter (which becomes smaller as the activity of price jumps increases), and it is not clear whether one can separately identify rough volatility and jumps from it. Theorem~\ref{thm:cf}, on the other hand, is derived for a general rough process with jumps (and jump intensities possibly varying in a rough way) and importantly does not impose any restriction on the Hurst parameter. With this result, as we shall show in Section \ref{sec:application} below, we can answer in the affirmative the question whether rough volatility and jumps can be separately identified from option observations in a general nonparametric setting.
\end{remark}

We make several comments about the  expansion result in Theorem~\ref{thm:cf}. First, compared with the corresponding result in \cite{T21} for the standard case when the spot semimartingale characteristics do not have rough components in their paths and $x$ does not have jumps of infinite variation, we note that the higher-order terms in the expansion here are of much higher asymptotic order. This underscores the importance of deriving such a higher-order expansion in a rough setting. The leading term of $|\call_{t,T}(u)|$ is $\exp (-\frac12 u^2\int_0^1 \si^2_{t+sT\mid t}ds )$, which is $O(1)$, while that of the argument of the characteristic function is the highest among $\Im(T\psi_t(\tfrac {u}{\sqrt{T}}))$, $-u^3\frac12\si^2_t\eta_t^v T^{1/2}$ and $-u^3 {\si_t^2\si^v_t}\Ga(H+\frac{5}{2})^{-1}T^H$. These three terms are of asymptotic order $O(T^{1-r/2})$, $O(\sqrt{T})$ and $O(T^H)$, respectively. Therefore, which of them dominates will depend on the degree of roughness of volatility and the degree of jump activity, and of course on which of these components are actually present in the dynamics of $x$. 

Similar comments apply for the higher-order terms in the asymptotic expansion due to the dynamics of the spot semimartingale characteristics. In particular, $\{C_{1,i}(u,T)_t\}_{i=1,\ldots,4}$   and $C'_{1,1}(u,T)_t$ are all due to the dynamics of the infinite jump variation component of $x$. The leading ones among them are $C_{1,1}(u,T)_t$ and $C'_{1,1}(u,T)_t$, both of which are  of order $O(T^{H_{\delta}+1-{r}/{2}})$ and depend on the rough component of the jump size function $\delta$. We further have $ C_{1,2}(u,T)_t + C_{1,3}(u,T)_t+C_{1,4}(u,T)_t = O(T^{ {3}/{2}- {r}/{2}})$, and these three terms do not depend on the rough component in $\delta$. 

The higher-order term in the expansion due to the rough component of the volatility dynamics is the term $C_2(u)_tT^{2H}$, which is of order $O(T^{2H})$. Comparing this contribution with the one due to the variation of $\delta$, which is of order $O(T^{H_{\delta}+1- {r}/{2}})$, if the small jumps have a rough component in their dynamics, we see that their asymptotic ranking in terms of size will depend on the values of the three parameters $H$,  $H_\delta$ and $r$. Of course, the different components in the asymptotic expansion of Theorem~\ref{thm:cf} also typically depend in a different way on $u$, and this can be used to separate their contribution to the conditional characteristic function as we will illustrate in our application of the expansion result.

We finish this section with presenting an expansion result for $\E_t[x_{t+T}-x_t]$, whose proof is almost immediate from \eqref{eq:x}. We will make use of this result later on when we estimate the degree of volatility roughness.
\begin{proposition}\label{prop:condmean}
	Suppose that the assumptions of Theorem~\ref{thm:cf} hold with $q=1$ and $H_\ga\geq H$. Then, as $T\to0$, we have that
	\[
	\E_t[x_{t+T}-x_t] = \biggl(\al_t+\int_\R \ga(t,z)\nu(dz)\biggr)T + O(T^{1+H}). \]
\end{proposition}

\section{Option-based characteristic function portfolios}\label{sec:options}

If $x$ is an asset price, the conditional characteristic function of its price increments can be inferred from portfolios of short-dated options, as done in \cite{QT19} and \cite{T19}. This   allows for feasible inference and the development of model specification tests on the basis of the expansion results of the previous section.

More specifically, under the simplifying assumption of zero risk-free rate and dividend yield, we have the following option spanning result from \cite{CM01}:
\begin{equation}\label{eq:spanning}
	\mathbb{E}_t[e^{iu(x_{t+T}-x_t)/\sqrt{T}}] = 1-\left(\frac{u^2}{T}+i\frac{u}{\sqrt{T}}\right)e^{-x_t}\int_{\mathbb{R}}e^{(iu/\sqrt{T}-1)(k-x_t)}O_{t,T}(k)dk
\end{equation}
for $u\in\R$,
where $O_{t,T}(k)$ denotes the price at time $t$ of an European style out-of-the-money option expiring at $t+T$ and with log-strike of $k$. In agreement with the notation used so far, the expectation above is taken under $\mathbb{Q}$, which signifies the so-called risk-neutral probability measure. 
We remind the reader that $O_{t,T}(k)$ is a put if $k\leq \log(F_{t,T})$ and a call otherwise, where $F_{t,T}$ is the time-$t$ futures price of the asset with expiration date $t+T$.

We can make a simple Riemann sum approximation of the integral in (\ref{eq:spanning}) using available options on a discrete log-strike grid. More specifically, for maturity $T_\ell$, we denote the available log-strike grid with 
\begin{equation}\label{eq:logmoney}
	\underline{k}_{\ell} ~\equiv~ k_{\ell,1} \, < \, k_{\ell,2} \, < \, \cdots \, < \, k_{\ell,N_\ell}~\equiv~ \overline{k}_{\ell},\qquad N_\ell\in\mathbb{N}_+.
\end{equation} 
The gap between the log-strikes is denoted $\Delta_{\ell,i} = k_{\ell,i} - k_{\ell,i-1}$ for $i=2,\dots,N_\ell$ and $\ell=1,2,\dots$ The log-strike grids need not be equidistant, that is, $\Delta_{\ell,i}$ may differ across $i$'s. For simplicity, in the above notation of the log-strike grid, we have dropped the dependence on $t$, as $t$ will be fixed throughout our applications. 

The true option prices are all defined on $(\Omega,\mathcal{F},(\mathcal{F}_t)_{t\geq 0},\P)$, where $\P$ is the statistical (true) probability measure and is locally equivalent to $\Q$. The observed option prices contain errors, that is,
\begin{equation}\label{eq:obs}
	\widehat{O}_{t,T_\ell}(k_{\ell,j}) = O_{t,T_\ell}(k_{\ell,j})+\epsilon_{t,T_\ell}(k_{\ell,j}),\quad j=1,\dots,N_\ell,\quad \ell=1,2,\dots,
\end{equation}
where the errors $\epsilon_{t,T_\ell}(k_{\ell,j})$ are defined on an auxiliary space 
$(\Omega^{(1)},\mathcal{F}^{(1)})$ equipped with a transition probability $\mathbb{P}^{(1)}(\omega,d\omega^{(1)})$ from $\Omega$, the probability space on which $x$ is defined, to $\Omega^{(1)}$. We further define
\[\ov\Omega \,=\, \Omega\times \Omega^{(1)},\quad\ov{\mathcal{F}} \,=\, \mathcal{F} \otimes \mathcal{F}^{(1)},\quad\ov{\mathbb{P}}(d\omega,d\omega^{(1)}) \,=~ \mathbb{P}(d\omega) \, \mathbb{P}^{(1)}(\omega,d\omega^{(1)}) \,.\]    

Using the available options, our estimator of $\call_{t,T}(u)$ from \eqref{eq:condchar} is given by 
\begin{equation}\label{eq:L_hat}
	\widehat{\mathcal{L}}_{t,T_\ell}(u) = 1 - \left(\frac{u^2}{T}+i\frac{u}{\sqrt{T}}\right)e^{-x_t}\sum_{j=2}^{N_\ell}e^{(iu/\sqrt{T_\ell}-1)(k_{\ell,j-1}-x_t)}\widehat{O}_{T_\ell}(k_{\ell,j-1})\Delta_{\ell,j},~u\in\mathbb{R}.
\end{equation}
Our estimation procedure in Section~\ref{sec:application} will be based on the argument of the characteristic function $\call_{t,T}(u)$. As seen from the expansion result in \eqref{eq:cf}, the value $\alpha_t$ of the spot drift term at time $t$  appears in it. We can estimate it from the option data. More specifically, similarly to (\ref{eq:spanning}), we have
\begin{equation}\label{eq:M}
	\mathbb{E}_t[x_{t+T}-x_t] = -\int_{\mathbb{R}}e^{-k}O_{t,T}(k)dk.
\end{equation}
The feasible (and standardized) counterpart of this is given by 
\begin{equation}\label{eq:M_hat}
	\widehat{\mathcal{M}}_{t,T_\ell} =- \frac{1}{T_\ell}\sum_{j=2}^{N_\ell}e^{-k_{\ell,j-1}}\widehat{O}_{T_\ell}(k_{\ell,j-1})\Delta_{\ell,j}.
\end{equation}
We note that $-2\widehat{\mathcal{M}}_{t,T_\ell}$ is an estimate of the conditional expected integrated variance of $x$ (this follows by an application of It\^o's lemma and the definition of the risk-neutral measure) and is used in the computation of the VIX volatility index by the CBOE options exchange. 

Our goal in this section is to derive a bound for the error $\widehat{\mathcal{L}}_{T_\ell}(u) - \mathcal{L}_{T_\ell}(u)$ and $\widehat{\mathcal{M}}_{t,T_\ell} - \frac{1}{T_\ell}\mathbb{E}_t[x_{t+T_\ell}-x_t]$ as the mesh of the log-strike grid, $\sup_{i=2,...,N_\ell}\Delta_{\ell,i}$, shrinks towards zero. For deriving such a result, we need several assumptions concerning existence of conditional moments of $x$, the observation scheme as well as the observation errors. They are stated below. As before, if expectation is taken under $\mathbb{Q}$, we will not use superscript in the notation; if expectation is under the true probability or $\ov \P$, we put superscripts to signify this.

\bass\label{ass:C} There are $\mathcal{F}_t$-measurable random variables $c_t$ and $C_t$ such that the following holds:
\benu
\item We have $\sigma_t>0$ and there exists   $\overline{t}>t$ such that for $s\in[t,\overline{t}]$ we have
\begin{equation}\label{a3:1}
	\begin{split}
		&\mathbb{E}_t[|\sigma_s|^4]+\mathbb{E}_t[e^{2|x_s|}] + \mathbb{E}_t \biggl[\biggl( \int_{\mathbb{R}} (e^{|\gamma(s,z)|}-1)\nu(dx)  \biggr)^2 \biggr]\\&\qquad+ \mathbb{E}_t\biggl[\biggl( \int_{\mathbb{R}} (e^{|\delta(s,z)|}-1-|\delta(s,z)|)\nu(dz)  \biggr)^2\biggr]<C_t.
	\end{split}
\end{equation}
\item The log-strike grids $\{k_{\ell,j}\}_{j=1}^{N_{\ell}}$, for $\ell=1,2,\dots$, are $\mathcal{F}_t$-measurable and we have
\begin{equation}\label{a4:1}
	c_t\Delta\leq k_{\ell,j} - k_{\ell,j-1}\leq C_t\Delta,\qquad \ell=1,2,\dots,
\end{equation}
where $\Delta$ is a deterministic sequence.
\item We have 
\begin{equation}\label{eq:epsilon} 
	\epsilon_{t,T_\ell}(k_{\ell,j}) = \zeta_{t,\ell}(k_{\ell,j})\overline{\epsilon}_{t,\ell,j}O_{t,T_\ell}(k_{\ell,j})
\end{equation}
for some $\mathcal{F}_t$-measurable $\zeta_{t,\ell}(k_{\ell,j})$ and some sequences of $\calf^{(1)}$-measurable random variables $\{\overline{\epsilon}_{t,\ell,j}\}_{j=1}^{N_\ell}$, for $\ell=1,2,\dots$, that are i.i.d.\  and independent of each other and of $ \mathcal{F}$ under $\P$. We further have 
\begin{equation}\label{eq:mom} 
	\mathbb{E}^{\ov \P}[\overline{\epsilon}_{t,\ell,j}\mid \mathcal{F}]= 0,\quad \mathbb{E}^{\ov \P}[(\overline{\epsilon}_{t,\ell,j})^2\mid \mathcal{F}] = 1,\quad \mathbb{E}^{\ov \P}[| \overline{\epsilon}_{t,\ell,j}|^{\kappa}\mid\mathcal{F}] <\infty
\end{equation} 
for some $\kappa\geq 4$ and all $\ell=1,2,\dots$
\eenu
\eass

Under Assumption \ref{ass:C}, we have 
\[\widehat{\mathcal{M}}_{t,T_\ell}\stackrel{\ov \P}{\longrightarrow} \frac1{T_\ell}\mathbb{E}_t[x_{t+T_\ell}-x_t]\quad\text{and}\quad \widehat{\mathcal{L}}_{T_\ell}(u) \stackrel{\ov \P}{\longrightarrow} \mathcal{L}_{T_\ell}(u),\]
locally uniformly in $u$   for all $\ell=1,2,\dots$ In the next theorem, we provide a bound for the rate of convergence of the feasible estimators. Recall that $a_n \asymp b_n$ if there are $c,C>0$ such that $c\lvert a_n\rvert\leq \lvert b_n\rvert\leq C\lvert a_n\rvert$ for all $n$.

\begin{theorem}\label{thm:err_bound} Suppose Assumptions \ref{ass:A}, \ref{ass:B} and \ref{ass:C} hold. For some finite $L\in\mathbb{N}_+$, assume that there are $\beta, \gamma>0$ and $\alpha>\frac{1}{2}$ such that as $T,\Delta\to 0$, $\un k_\ell\to-\infty$ and $\ov k_\ell\to\infty$, we have $T_\ell\asymp T$, $\Delta\asymp T^{\alpha}$, $e^{\underline{k}_\ell}\asymp T^{\beta}$, $e^{\overline{k}_\ell}\asymp T^{-\gamma}$ for all $\ell=1,\dots,L$. Then, if $\beta\wedge \gamma>\frac{\alpha}{4}-\frac{1}{8}$, we have 
	\begin{equation}\label{err:1}
		\Bigl\lvert\widehat{\mathcal{L}}_{t,T_\ell}(u) - \mathcal{L}_{t,T_\ell}(u)\Bigr\rvert+ \biggl\lvert\widehat{\mathcal{M}}_{t,T_\ell} - \frac{1}{T_\ell}\mathbb{E}_t[x_{t+T_\ell}-x_t]\biggr\rvert= O_{\ov \P}^\mathrm{uc}\left(\frac{\sqrt{\Delta}}{T^{1/4}}\right),\quad \ell=1,2,\dots
	\end{equation}
\end{theorem}
The rate of convergence of the option-based estimators is determined by the mesh of the log-strike grid, $\Delta$, and the asymptotic order of the maturities of the options used in the estimation, $T$. We note that for consistent estimation, we need $\Delta/\sqrt{T}\rightarrow 0$. The reason for this is that if this is not the case the Riemann sum approximation error around-the-money, i.e., for log-strikes around $x_t$, will be too big relative to the quantity that is estimated (the conditional characteristic function or the conditional mean). In practice, this condition is expected to hold for the typical available strike grids and maturities, and this is reflected in the fact that the observed changes in option prices across consecutive strikes are typically not very big for options with at least one day to expiration. 

In the statement of the theorem, we further impose the requirement $\beta\wedge \gamma>\frac{\alpha}{4}-\frac{1}{8}$. This requirement guarantees that the error in $\widehat{\mathcal{L}}_{T_\ell}(u)$ and $\widehat{\mathcal{M}}_{t,T_\ell}$ due to the fact that Riemann sum is over a bounded (but asymptotically increasing) log-strike domain. This error is not expected to be binding in a typical application because we typically observe deep out-of-the-money option prices until their value falls below the minimum tick size.  
Finally, we note that with a slight strengthening of Assumption~\ref{ass:C}, we can also derive  a CLT for $\widehat{\mathcal{L}}_{t,T_\ell}(u)$ and $\widehat{\mathcal{M}}_{t,T_\ell}$. For brevity, we do not present such a result here. 

\section{Application: Estimating the degree of volatility roughness}\label{sec:application}

We will show in this section how one can estimate the volatility roughness parameter $H$ using the expansion result of Theorem~\ref{thm:cf} and short-dated options. The parameter $H$ enters both in the norm and the argument of the conditional characteristic function. Estimating $H$ from the latter appears somewhat easier provided one knows that $\si^v_t\neq 0$. Note that the process $\si^v$, together with $\eta^v$, captures the dependence between the price and volatility diffusive moves, and there is a lot of empirical evidence for the existence of such dependence. Therefore, an assumption that $\si^v_t\neq 0$ is not restrictive from an empirical point of view and we will work under such an assumption in this section.

\subsection{The case without jumps in $x$}
We first start by discussing the case where there are no jumps in $x$, that is,
\begin{equation}\label{eq:nojumps} 
	\ga(t,z)\equiv\delta(t,z)\equiv0.
\end{equation}
We might still have jumps in $\si$, $\eta$ or $\wt \eta$, of course. 
For $u>0$, we introduce
\beq\label{eq:arg-def} \begin{split}
	A_{t,T}(u) &= \Arg(\mathcal{L}_{t,T}(u)) - \frac{u}{\sqrt{T}}\mathbb{E}_t[x_{t+T}-x_t],\\  \widehat{A}_{t,T}(u) &= \Arg(\widehat{\mathcal{L}}_{t,T}(u)) -u\sqrt{T}\widehat{\mathcal{M}}_{t,T} ,
\end{split}
\eeq
where as usual $\Arg(z)$ is the principal argument of the complex number $z$. Based on the expansions in Theorem~\ref{thm:cf} and Proposition~\ref{prop:condmean}, we derive the following result.

\begin{corollary}\label{cor:arg} Under \eqref{eq:nojumps} and the assumptions of Proposition~\ref{prop:condmean}, we have
	\begin{equation}\label{eq:arg}  
		A_{t,T}(u)   = -  \frac{\si_t^2f'(v_t)\si^v_tu^3}{\Ga(H+\frac{5}{2})}T^H  + O^\mathrm{uc}(T^{\frac12\wedge 2H}),\qquad T\to0. 
	\end{equation}
	If furthermore $\eta^v\equiv0$, then this expansion improves to
	\begin{equation}\label{eq:arg:imp} 
		A_{t,T}(u)   = -  \frac{\si_t^2f'(v_t)\si^v_tu^3}{\Ga(H+\frac{5}{2})}T^H  + O^\mathrm{uc}(T^{2H}),\qquad T\to0. 
	\end{equation}
\end{corollary}

To make use of the corollary, we pick two short tenors $0<T_1<T_2$ and denote $\tau = T_2/T_1$. Then \eqref{eq:arg} (resp., \eqref{eq:arg:imp}) implies that for any fixed $u>0$,
\[ \frac{A_{t,T_2}(u)}{A_{t,T_1}(u)} = \tau^{H}+O(T_1^{(\frac12-H)\wedge H})\quad \text{(resp., $O(T_1^{H})$)}, \]
from which one can easily estimate $H$ with the same rate. To improve efficiency, we can use multiple $u$'s in the estimation. We can also perform estimation on the basis of certain transformations of $A_{t,T}(u)$ that can help reduce finite-sample biases. More specifically, in a first step, we rewrite \eqref{eq:arg} (resp., \eqref{eq:arg:imp}) as
\[	A_{t,T}(u)   = -  \frac{\si_t^2f'(v_t)\si^v_tu^3}{\Ga(H+\frac{5}{2})}T^H  + C_{1,5}(T)_t u^5+ O^\mathrm{uc}(T^{\frac12\wedge 2H})\quad \text{(resp., $O^\mathrm{uc}(T^{2H})$)},\]
where $C_{1,5}(T)_t=4\sum_{j=1}^\infty \al_j^3\beta_j^2$ in the notation of \eqref{eq:cfXT-2} in the proof of Theorem~\ref{thm:cf}. While $C_{1,5}(T)_t=O(T^{3H})$, including it in the last display improves the performance of our estimators for given $T_1$ and $T_2$ in the simulation study. Also, we find that taking the inverse of $A_{t,T}(u)$ is more preferable, leading us to the expansion
\begin{align*}
	-\frac{u^3}{A_{t,T}(u)}&=\frac{1}{\si_t^2f'(v_t)\si^v_t\Ga(H+\frac52)^{-1}T^H-C_{1,5}(T)_t u^2 }+ O (T^{(\frac12-2H)\wedge 0})\quad \text{(resp., $O(1)$)}\\
	&=\frac{\Ga(H+\frac52)}{\si_t^2f'(v_t)\si^v_t}T^{-H} + \frac{\Ga(H+\frac52)^2C_{1,5}(T)_t}{\si_t^4(f'(v_t)\si^v_t)^2T^{2H}}u^2+ O (T^{(\frac12-2H)\wedge 0})\quad \text{(resp., $O(1)$)}.
\end{align*}
Taking a grid of points $\underline{u} = (u_1,\dots,u_K)$ for $0<u_1<\dots<u_K$ and some integer $K\geq 2$, we then define
\begin{equation}\label{eq:unA} 
	A_{t,T}(\underline{u}) = \frac{ \sum_{i=1}^Ku_i^2\sum_{i=1}^K A_{t,T}(u_i)^{-1}u_i^5-\sum_{i=1}^Ku_i^4\sum_{i=1}^K u_i^3 A_{t,T}(u_i)^{-1}    }{ K\sum_{i=1}^Ku_i^4 - \bigl(\sum_{i=1}^Ku_i^2\bigr )^2  },
\end{equation}
which is the intercept of a regression of $-u_i^3/A_{t,T}(u_i)$ on a constant and $u_i^2$. From the preceding discussion and the mean-value theorem, we obtain 
\begin{equation}\label{eq:ratio} 
	\frac{A_{t,T_1}(\underline u)}{A_{t,T_2}(\underline u)} = \tau^{H}+O(T_1^{(\frac12-H)\wedge H})\quad \text{(resp., $O(T_1^{H})$)}. 
\end{equation}
For a feasible estimator on the basis of the above result, we define $\wh A_{t,T_j}(\underline u)$, for $j=1,2$, in the same way as $A_{t,T_j}(\underline u)$ in \eqref{eq:unA} but with $\wh A_{t,T_j}( u_i)$ instead of $A_{t,T_j}(u_i)$. Then, the estimator of $H$ is given by 
\begin{equation}\label{eq:Hest} 
	\wh H_n = \frac{\log \wh A_{t,T_1}(\underline u) - \log \wh A_{t,T_2}(\underline u)}{\log \tau}.
\end{equation}
The following theorem derives a bound for its rate of convergence.
\begin{theorem}\label{thm:Hest}
	Assume \eqref{eq:nojumps} and the hypotheses of Proposition~\ref{prop:condmean} and Theorem~\ref{thm:err_bound} hold. If $\si_t>0$ and $\si^v_t\neq 0$, then 
	\begin{equation}\label{eq:H-conv} 
		\wh H_n = H+O_{\ov \P}(T_1^{(\frac12-H)\wedge H}  \vee  \Delta^{\frac12}T_1^{-\frac14-H}).
	\end{equation}
	If $\eta^v\equiv0$ in addition, this improves to
	\begin{equation}\label{eq:H-conv-3} 
		\wh H_n = H + O_{\ov \P}(T_1^{H} \vee   \Delta^{\frac12}T_1^{-\frac14-H}).
	\end{equation}
	If $\si_t>0$, $\eta\equiv0$ but $\eta_t^v\neq 0$, we have
	\begin{equation}\label{eq:H-conv-2} 
		\wh H_n = \frac{1}{2} + O_{\ov \P}(T_1 \vee   \Delta^{\frac12}T_1^{-\frac34}).
	\end{equation}
\end{theorem}

The last result in (\refeq{eq:H-conv-2}) shows that if $\sigma$ does not contain a rough component, then our estimator converges to $\frac12$, which is the value of the Hurst parameter of the Brownian motion.

\subsection{The case with jumps in $x$}
Next, we shall describe how one can construct an estimator of $H$ in the presence of jumps that  retains the same rate of convergence. For this, we need to take into account the fact that $A_{t,T}(u)$ is affected by   jumps. To remove the effect due to them, we will use $A_{t,T}(u)$ for different $T$'s and we will suitably difference them. This will eliminate the leading term in $A_{t,T}(u)$ due to the jumps but will result in a loss of signal for the value of $H$. This is the price to pay for the robustness of the estimation of $H$ to jumps in $x$. For realistic time-to-maturity values, we expect such a jump-robust estimator of $H$ to perform relatively poor due to the reduced signal about $H$ and the consequent high sensitivity to even small finite-sample biases in the estimation. The situation here is reminiscent of the problem of estimating multiple jump activity indices, see \cite{jacod2016efficient}. Nevertheless, our subsequent analysis shows that the roughness parameter $H$ can be recovered consistently, even in the presence of infinite variation jumps.

To fix ideas, let us pick four (instead of two) short tenors $0<T_1<T_2<T_3<T_4$. For $j=2,3,4$, let  $\tau_j = T_j/T_1$ and 
\beq\begin{split}
	A_{t,T_1,T_j}(u) &= \frac{1}{u^3} \bigl(\Arg(\call_{t,T_1}(u)) -\Arg(\call_{t,T_j}(u\sqrt{\tau_j}))/\tau_j\bigr),\\
	\wh A_{t,T_1,T_j}(u) &= \frac{1}{u^3} \bigl(\Arg(\wh \call_{t,T_1}(u)) -\Arg(\wh\call_{t,T_j}(u\sqrt{\tau_j}))/\tau_j\bigr).
\end{split}
\eeq 
With this differencing approach, the contribution of the drift $\al_t$ and of the jumps through $T\psi_t(u/\sqrt{T})$ to the argument of $\call_{t,T}(u)$ is canceled out identically. In order to further remove $C_{1,j}(u,T)_t$ where $j=2,3,4$ as well as the term $-\frac12 u^3\si_t^2\eta_t^v T^{1/2}$ in \eqref{eq:cf}, we consider second-order differences of the form 
\beq\label{eq:secdiff}\begin{split}
	A_{t,T_1,T_j,T_4}(u) &= A_{t,T_1,T_j}(u)-\frac{\tau_j-1}{\tau_4-1} A_{t,T_1,T_4}(u), \\
	\wh A_{t,T_1,T_j,T_4}(u) &= \wh A_{t,T_1,T_j}(u)-\frac{\tau_j-1}{\tau_4-1} \wh A_{t,T_1,T_4}(u),\qquad j=2,3.
\end{split}
\eeq 
\begin{theorem}\label{thm:secdiff}
	Under the assumptions of Theorem~\ref{thm:cf}, we have 
	\begin{equation}\label{eq:threeT}  
		A_{t,T_1,T_2,T_4}(u) = \frac{\si_t^2f'(v_t)\si^v_t}{\Ga(H+\frac52)} \biggl[(\tau_2^{\frac12+H}-1) - \frac{\tau_2-1}{\tau_4-1} (\tau_4^{\frac12+H}-1)\biggr]T_1^H  + O(T_1^{(H_\delta+1-\frac r2)\wedge 2H}).
	\end{equation}
	Moreover, if the conditions of Theorem~\ref{thm:err_bound} are satisfied and $\si_t>0$ and $\si^v_t\neq 0$, then we have the following result for any $u>0$: if $\tau_2$, $\tau_3$ and $\tau_4$ are such that
	the mapping $$ F(H)=\biggl[(\tau_2^{\frac12+H}-1) - \frac{\tau_2-1}{\tau_4-1} (\tau_4^{\frac12+H}-1) \biggr] \Big/ \biggl[(\tau_3^{\frac12+H}-1) - \frac{\tau_3-1}{\tau_4-1} (\tau_4^{\frac12+H}-1) \biggl]$$
	has an inverse $F^{-1}$ and $\frac{d}{dH} F(H)\neq0$, then
	the estimator
	\begin{equation}\label{eq:Hn'} 
		\wh H'_n= F^{-1}\biggl(\frac{\wh A_{t,T_1,T_3,T_4}(u)}{\wh A_{t,T_1,T_2,T_4}(u)}\biggr)
	\end{equation}
	satisfies
	\begin{equation}\label{eq:H-conv-4} 
		\wh H'_n= H+O_{\ov\P}(T_1^{(H_\delta-H+1-\frac r2)\wedge H} \vee \Delta^{\frac12}T_1^{-\frac14-H}).
	\end{equation}
\end{theorem}

As seen from the result of the theorem, $\wh H'_n$ can achieve the same rate of convergence as $\wh H_n$ even when jumps are present in $x$. The loss of information due to the robustification for jumps is easy to illustrate. Suppose, $\tau_2 =2$, $\tau_3 = 3$ and $\tau_4 = 4$ as for a typical set of available option maturities. In this case, $F(0) = 1.2370$ and $\textrm{lim}_{H\rightarrow 1/2}F(H) = 1.1525$. This implies a relatively small variation in the function $F(H)$ as $H$ varies in $(0,\frac12]$. By contrast, the counterpart of $F(H)$ for the estimator $\wh H_n$ is $\tau_2^{H}$ and this quantity varies from $1$ for $H=0$ to $1.4142$ for $H=\frac12$, which is a much larger spread.    

We note that the estimator $\wh H'_n$ should be used only if one knows that the volatility has a rough component, that is,  $\si^v_t \neq 0$. If this is not the case, one can first test for rough volatility (in the presence of jumps) using the ratio $\wh A_{t,T_1,T_2}(u)/\wh A_{t,T_1,T_3}(u)$. We implement such a test in our numerical experiments.

\section{Numerical experiments}\label{sec:mc}

In this section we will conduct several numerical experiments to assess the accuracy of the higher-order asymptotic expansion and some of its applications. 

\subsection{Setup} 
We will use the following rough volatility model for the asset price $X$:
\begin{equation}\label{eq:mc_rv}
	\frac{dX_t}{X_t} = \sqrt{V_t}dW_t,~~V_t = V_0 + \frac{\nu}{\Gamma (H + \frac{1}{2} )}\int_0^t(t-s)^{H-\frac{1}{2}}\sqrt{V_s}(\rho dW_s + \sqrt{1-\rho^2}d\widetilde{W}_s),
\end{equation}
where $W$ and $\widetilde{W}$ are independent Brownian motions, $V_0$ is an $\mathcal{F}_0$-measurable positive random variable, and $\nu>0$ and $\rho\in[-1,1]$ are some constants. This is the model considered by \cite{el2019roughening} and can be viewed as the limit of the rough Heston model of \cite{el2018microstructural} with zero mean-reversion. We generate option prices from the model using the codes from the paper of \cite{romer2022empirical}. 

We note that in this model, there are several restrictions to the general specification of $x = \log X$ in (\ref{eq:x}) and (\ref{eq:sigma})--(\ref{eq:delta}). First, there are no jumps, so that $\gamma(s,z) = \delta(s,z) = \ga_v(s,z) = \delta_\sigma(s,z) = 0$. Second, the variance only has  a rough component and this implies $\eta_s^v = \widetilde{\eta}_s^v = \overline{\eta}_s^v = 0$. These restrictions allow us to strengthen the asymptotic expansion result of Theorem~\ref{thm:cf} to the following one for  models with the above-mentioned restrictions:
\beq\label{eq:cf_rough} 
\call_{0,T}(u)=\exp\biggl(iu\al_0\sqrt{T}-\frac12 u^2\si^2_{0} -iu^3 \frac{\si_0^2f'(v_0)\si^v_0}{\Ga(H+\frac{5}{2})}T^H   \biggr)+C_2(u)_0T^{2H} +o^{\mathrm{uc}}(T^{2H}).
\eeq 

We set the parameters of the model as follows. The correlation parameter is set to $\rho = -0.9$, which allows for strong negative dependence between price and volatility innovations (also known as leverage effect). We put the volatility of volatility parameter to $\nu = 0.5$. This generates fat-tailed return distributions and consequently numbers of short-tenor options (for the empirically calibrated option observation scheme explained below) that are roughly consistent with those observed for market index options. Finally, we experiment with four values of the Hurst parameter: $H=0.1$, $0.25$, $0.4$ and $0.49$. 

In addition to the above rough volatility model, we also consider the following model with jumps in which volatility follows the standard Heston square-root diffusion dynamics:
\begin{equation}\label{eq:mc_jd}
	\begin{split}
		\frac{dX_t}{X_t} &= \sqrt{V_t}dW_t+\int_{\mathbb{R}}(e^x-1)\widetilde{\mu}(dt,dx),\\ V_t &= V_0 + \int_0^t(t-s)\kappa(\theta-V_s)ds+\int_0^t 0.5\sqrt{V_s}(\rho dW_s + \sqrt{1-\rho^2}d\widetilde{W}_s),
	\end{split}
\end{equation}  
and where $\mu$ is an integer-valued random measure with compensator $dt\otimes \nu_t(dx)$, for $\nu_t$ given by
\begin{equation}
	\nu_t(dx) = V_t\times\nu^{vg}(x)dx,~\nu^{vg}(x) = c_{-}\frac{e^{-\lambda_-|x|}}{|x|}1_{\{x<0\}} + c_{+}\frac{e^{-\lambda_+|x|}}{|x|}1_{\{x>0\}}.
\end{equation} 
The jumps in $X$ is a time-changed variance gamma process, see \cite{CGMY03}, with the time-change being the integrated diffusive variance. This specification belongs to the exponentially-affine class of models of \cite{DPS00} that has been commonly used in empirical parametric option pricing. As for the rough volatility specification, we set $\rho=-0.9$ and $\nu = 0.5$. We further set the mean of volatility parameter to $\theta = 0.03$ and the mean-reversion parameter to $\kappa = 8$. The latter corresponds to half-life of a shock to volatility of one month. 

Turning next to the jump specification, we set $\lambda_- = 30$ and $\lambda_+ = 100$. These values imply decay in deep out-of-the-money short-dated option prices similar to those observed for S\&P 500 index options. We then set the scale parameters $c_{\pm}$ to  
\begin{equation}  
	c_- = 0.3\times \lambda_-^{2}~\textrm{and}~c_+ = 0.2\times \lambda_+^{2}.
\end{equation}
For this choice of $c_{\pm}$, we have that $\int_{x<0}x^2\nu_t(dx) = 0.3\times V_t$ and $\int_{x>0}x^2\nu_t(dx) = 0.2\times V_t$. Therefore, jump variance is half of the diffusive variance and negative jumps have higher variation than positive jumps. 

\subsection{First-order versus higher-order asymptotic expansion}

We start with assessing the improvement from the developed higher-order asymptotic expansion over the first-order one given in (\ref{eq:fo_exp}). The latter implies the following estimate of spot volatility from the characteristic function:
\begin{equation}
	V_{t,T}(u) = -\frac{2}{u^2}\log|\mathcal{L}_{t,T}(u)|,
\end{equation}
for any finite $u>0$ and small $T>0$. Higher values of $u$ are necessary in order to make $V_{t,T}(u)$ robust to jumps in $x$. Indeed, the limit of $V_{t,T}(u)$ as both $T\downarrow 0$ and $u\downarrow 0$ is the predictable spot quadratic variation of $x$. Given our higher-order asymptotic expansion, we know that $V_{t,T}(u)$ contains a bias due to the time variation in volatility, i.e., the term $C_2(u)_tT^{2H}$ in (\ref{eq:cf}). This bias can be removed by using characteristic functions at the same time $t$ and for two different small values of $T$, say, $T_1<T_2$. More specifically, based on our higher-order asymptotic expansion, we have the following improved estimator of spot volatility that does not contain the bias  $C_2(u)_tT^{2H}$:
\begin{equation}
	V_{t,T_1,T_2}(u,H) = V_{t,T_1}(u) - T_1^{2H}\frac{V_{t,T_2}(u) -V_{t,T_1}(u) }{T_2^{2H}-T_1^{2H}}.
\end{equation} 
We can, therefore, assess the improvement offered by the higher-order expansion by simply checking how much closer is $V_{t,T_1,T_2}(u,H)$ to $\int_0^1\sigma_{t+sT|t}^2ds$ than $V_{t,T_1}(u)$. We will use the rough volatility specification in (\ref{eq:mc_rv}) for this. We set $t=0$ so that $\int_0^1\sigma_{t+sT|t}^2ds = V_0$ and we consider $V_0 = 0.03$. We will further set $T_1 = 1/252$ and $T_2 = 2/252$, which correspond to horizons of $1$ and $2$ business days, with a unit of time being one calendar year. 

The results are reported in Figure~\ref{fig:v_comp}. On each plot, the maximum value of $u$ is set to $\inf\{u\geq 0: |\mathcal{L}_{0,T_1}(u)|\leq 0.25\}$. The different range of values of $u$ for the different cases of $H$, particularly for $H=0.1$,  show indirectly the nontrivial impact   rough volatility dynamics have on $\mathcal{L}_{0,T}(u)$, even for a value of $T$ as small as one day. The reported results in Figure~\ref{fig:v_comp} show that in general $V_{t,T_1,T_2}(u,H)$ provides nontrivial reduction in the downward bias in $V_{0,T_1}(u)$ due to the (rough) volatility dynamics. These improvements tend to be bigger for larger values of $u$, which is consistent with the fact that the term $C_2(u)$ increases with $u$. We stress, however, that our asymptotics is for $u$ fixed. For larger values of $u$, even $V_{t,T_1,T_2}(u,H)$ starts having a nontrivial bias. This effect is stronger for lower values of $H$. In fact, for the lowest considered value of $H=0.1$,  $V_{t,T_1,T_2}(u,H)$ works well only for relatively small levels of $u$.

\begin{figure}[htbp]
	\begin{center}
		\includegraphics[width=\textwidth]{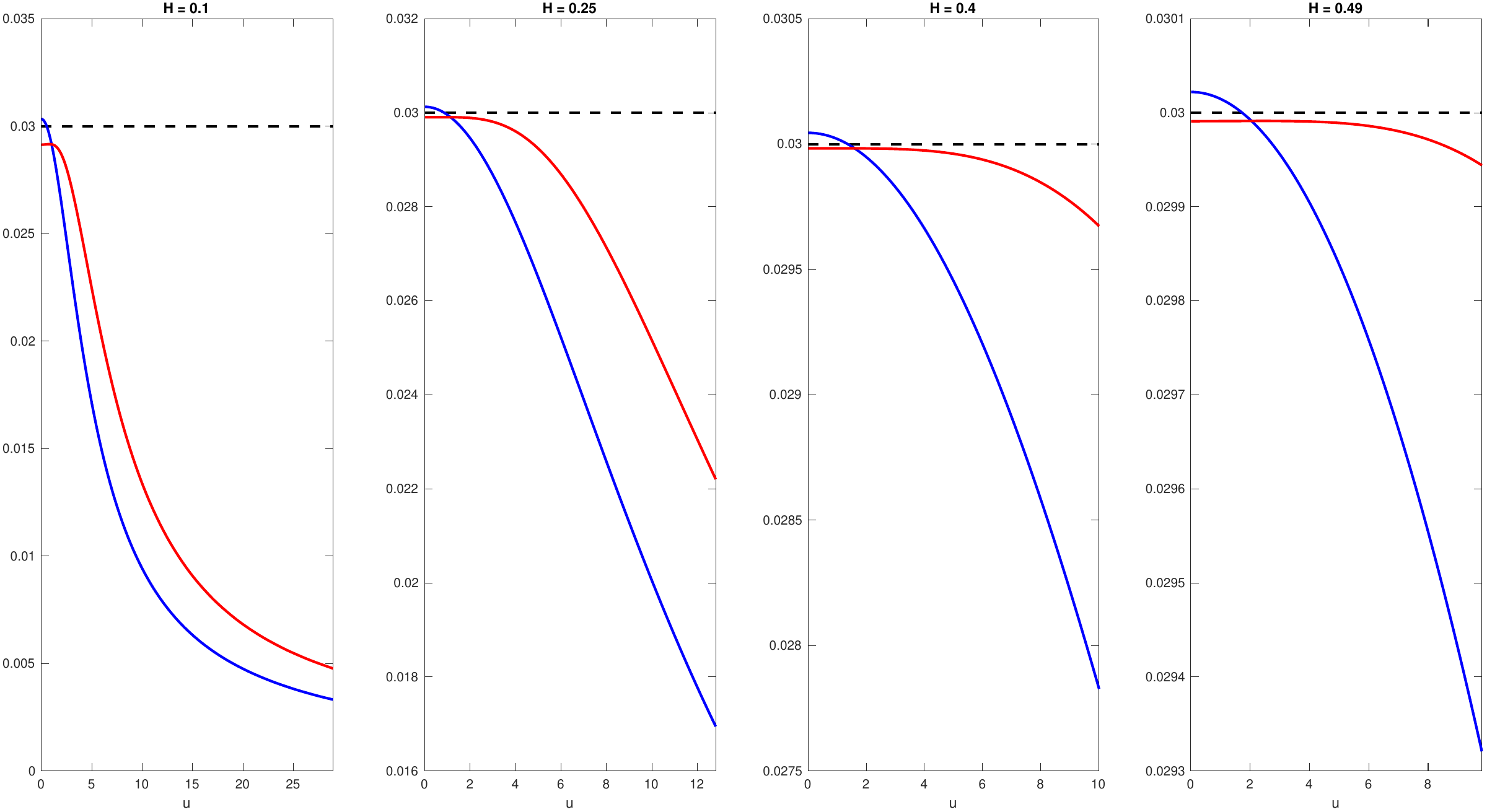}
	\end{center}
	\caption{$V_{0,T_1}(u)$ and $V_{0,T_1,T_2}(u,H)$. Dashed horizontal line is at $V_0$, blue line is $V_{0,T_1}(u)$ and red line is $V_{0,T_1,T_2}(u,H)$. Maximum value of $u$ on each plot is given by $\inf\{u\geq 0: |\mathcal{L}_{0,T_1}(u)|\leq 0.25\}$.}
	\label{fig:v_comp}
\end{figure}

\subsection{Monte Carlo study}
We next conduct a Monte Carlo evaluation of inference procedures from short-dated options related to the roughness of volatility. We start with describing the option observation scheme. We consider three short tenors: $T_1 = 1/252$, $T_2 = 2/252$ and $T_3 = 3/252$, corresponding to $1$, $2$ and $3$ business days to expiration. The strike grid of observed options is equidistant at increments of $5$. Starting from the current stock price, we keep adding lower and higher strikes until the true option prices fall in value below $0.075$. This setup mimics that of available options written on the S\&P 500 index. Finally, option prices are observed with error that is proportional to the true option price. More specifically, we have
\begin{equation}
	\widehat{O}_{0,T_\ell}(k_{\ell,j}) = (1+0.025\times \epsilon_{t,T_\ell}(k_{\ell,j}))O_{0,T_\ell}(k_{\ell,j}),\quad \ell=1,2,
\end{equation}    
where $\{\epsilon_{t,T_1}(k_{1,j}\}_{j\geq 1}$ and $\{\epsilon_{t,T_2}(k_{2,j}\}_{j\geq 1}$ are two independent i.i.d. sequences of standard normal random variables. We note that, since the option price varies a lot across strikes for a given tenor (in fact, the option price asymptotic rate of decay differs depending on how far the strike is from the current price), the above option observation model implies nontrivial heteroskedasticity in the option observation errors. 

In all considered cases, we set the starting value of the price at $X_0 = 3500$. For the starting variance, we consider three values: $V_0 = 0.015$ (low volatility), $V_0 = 0.03$ (medium volatility) and $V_0 = 0.06$ (high volatility). These numbers are calibrated to the level of volatility of the S\&P 500 market index.

\subsubsection{Inference assuming no jumps in $x$}
We start with evaluating the performance of the estimator $\widehat{H}_n$ of the degree of roughness of volatility that assumes knowledge that jumps are not present in $x$. For implementing this estimator, we need to set the grid of points of $u$ on which we evaluate $\widehat{A}_{0,T_1}(u)$ and $\widehat{A}_{0,T_2}(u)$. We determine this set in an adaptive way. More specifically, we set 
\begin{equation}\begin{split}
		\widehat{u}_1 &= \inf\{u\geq 0: |\widehat{\mathcal{L}}_{0,T_2}(u)|<0.9\},\quad \widehat{u}_K = \inf\{u\geq 0: |\widehat{\mathcal{L}}_{0,T_2}(u)|<0.75\},\\
		\widehat{u} &= \inf\{u\geq 0: |\widehat{A}_{0,T_2}(u)|>\pi/2\},\end{split}
\end{equation}
and from here $\widehat{\underline{u}} = \{\widehat{u}_1\wedge\widehat{u}:0.01:\widehat{u}_K\wedge \widehat{u}\}$. Note that the probability that $\widehat{\underline{u}}$ contains less than two points converges to zero. This also never happens in our simulations and for this reason we do not consider the situation when this is not the case as this is irrelevant from an asymptotic point of view.

The results from the Monte Carlo are reported in Table~\ref{tb:mc_results_nj}. Overall, they indicate good finite sample performance of the estimator across the different volatility regimes and the different values of $H$. In most of the cases, the bias in the estimator is small. The bias is only slightly bigger in relative terms in the case $H=0.1$ when the starting value of volatility is $V_0 = 0.015$ or $V_0 = 0.030$. This is likely due to the fact that this very rough case produces higher prices of deep out-of-the-money options and we keep only options which exceed in value $0.075$. Hence the error in $\widehat{\mathcal{L}}_{t,T}(u)$ due to the truncation of the limits of integration of the integral in (\ref{eq:spanning}) is bigger in this situation. For all other values of $H$, the estimator is not very sensitive to the different starting values of volatility. We note further that in all considered cases the interquartile range is very small, which indicates good precision. 
\begin{table}[htbp!]
	\caption{Monte Carlo results: Performance of $\widehat{H}_n$ based on $5{,}000$ Monte Carlo replications}
	\begin{tabular*}{0.6\textwidth}{@{\extracolsep{\fill}} cccrrr}
		\hline
		{\bf Case} & & & \multicolumn{3}{c}{Empirical Quantiles of $\widehat{H}_n$}\\
		& & & $25$-th~~ & $50$-th~~ & $75$-th~~\\[1.0ex]
		& & & \multicolumn{3}{c}{Low Volatility}\\[1.0ex]
		$V_0 = 0.015$, $H=0.10$  & & & $0.0327$  & $0.0660$ & $0.0968$\\
		$V_0 = 0.015$, $H=0.25$  & & & $0.2690$  & $0.2775$ & $0.2859$\\
		$V_0 = 0.015$, $H=0.40$  & & & $0.3970$  & $0.4140$ & $0.4312$\\
		$V_0 = 0.015$, $H=0.49$  & & & $0.4551$  & $0.4830$ & $0.5110$\\
		$V_0 = 0.015$, not rough  & & & $-0.6119$  & $-0.5689$ & $-0.5311$\\[1.0ex] 
		& & & \multicolumn{3}{c}{Medium Volatility}\\[1.0ex]
		$V_0 = 0.030$, $H=0.10$  & & & $0.0488$  & $0.0579$ & $0.0674$\\
		$V_0 = 0.030$, $H=0.25$  & & & $0.2593$  & $0.2683$ & $0.2773$\\
		$V_0 = 0.030$, $H=0.40$  & & & $0.3927$  & $0.4125$ & $0.4325$\\
		$V_0 = 0.030$, $H=0.49$  & & & $0.4434$  & $0.4777$ & $0.5108$\\
		$V_0 = 0.030$, not rough  & & & $-0.3548$  & $-0.3438$ & $-0.3330$\\[1.0ex] 
		& & & \multicolumn{3}{c}{High Volatility}\\[1.0ex]
		$V_0 = 0.060$, $H=0.10$  & & & $0.1046$  & $0.1103$ & $0.1155$\\
		$V_0 = 0.060$, $H=0.25$  & & & $0.2516$  & $0.2617$ & $0.2715$\\
		$V_0 = 0.060$, $H=0.40$  & & & $0.3864$  & $0.4097$ & $0.4321$\\
		$V_0 = 0.060$, $H=0.49$  & & & $0.4414$  & $0.4800$ & $0.5190$\\
		$V_0 = 0.060$, not rough  & & & $-0.2577$  & $-0.2484$ & $-0.2392$\\
		\hline
	\end{tabular*}
	\label{tb:mc_results_nj}
\end{table}

Finally, we report also the performance of the estimator $\widehat{H}_n$ for the model in (\ref{eq:mc_jd}) in which volatility is not rough and $x$ contains jumps. As seen from the reported results in  Table~\ref{tb:mc_results_nj}, the presence of jumps affects in a nontrivial way the performance of $\widehat{H}_n$. In absence of jumps, this estimator should converge to $0.5$. Instead, in  the model in (\ref{eq:mc_jd}) with jumps, this estimator has negative median for all three starting values of volatility. This, of course, is in line with our theoretical results and illustrates the difficulty in disentangling rough volatility from jumps.

\subsubsection{Inference allowing for jumps in $x$}

We now look at the situation when there can be jumps in $x$. We only analyze the behavior of a statistic that allows us to assess whether volatility is rough in presence of jumps in $x$. Recall that this is a necessary preliminary step for a jump-robust estimator of $H$.  For this, we will compare $\widehat{A}_{0,T_1,T_2}(u)$ and $\widehat{A}_{0,T_1,T_3}(u)$. Similarly to the case of no jumps in $x$, $u$ will be determined in an adaptive way. More specifically, we now set 
\begin{equation}\begin{split}
		\widehat{u}_1 &= \inf\{u\geq 0: |\widehat{\mathcal{L}}_{0,T_3}(u)|<0.75\},\\
		 \widehat{u}_K& = \inf\{u\geq 0: |\widehat{\mathcal{L}}_{0,T_3}(u)|<0.25\},\\
		\widehat{u} &= \inf\{u\geq 0: |\widehat{A}_{0,T_3}(u)|>\pi/2\},\\
		\overline{u} &= \sup\{u\geq \widehat{u}_1: \textrm{sign}(\widehat{A}_{0,T_1,T_2}(u)) = \textrm{sign}(\widehat{A}_{0,T_1,T_2}(\widehat{u}_1))  \\&\quad~~~~~~~~~~~~~~~~~~~~~~~~~~~=\textrm{sign}(\widehat{A}_{0,T_1,T_3}(u)) = \textrm{sign}(\widehat{A}_{0,T_1,T_3}(\widehat{u}_1))\},
	\end{split}
\end{equation}
and from here
\begin{equation}\widehat{\underline{u}} = \{\widehat{u}_1\wedge\widehat{u}\wedge\overline{u}:0.01:\widehat{u}_K\wedge \widehat{u}\wedge\overline{u}\}.
\end{equation} Compared to before, we determine $\widehat{u}_1$, $\widehat{u}_K$ and $\widehat{u}$ from the third tenor and we have added $\overline{u}$ in determining $\widehat{\underline{u}}$. The condition involving $\overline{u}$ is a requirement that  we evaluate $\widehat{A}_{0,T_1,T_2}(u)$ and $\widehat{A}_{0,T_1,T_3}(u)$ at points where  they have the same sign. This should be always the case asymptotically given our expansion result. 

Our statistic that discriminates between the case of rough volatility and no rough volatility is given by
\begin{equation}
	\widehat{T}_n = \log\left(\frac{1-\tau_3}{1-\tau_2}\frac{\sum_{u\in \widehat{\underline{u}}} \widehat{A}_{0,T_1,T_2}(u)}{\sum_{u\in \widehat{\underline{u}}} \widehat{A}_{0,T_1,T_3}(u)}\right).
\end{equation}
Given our expansion results of the characteristic function, we know that 
\begin{equation}
	\widehat{T}_n~\stackrel{\P}{\longrightarrow}~ \begin{cases}  0 & \textrm{if volatility is not rough},\\ C=C(H)>0 & \textrm{if volatility is rough}.\end{cases}
\end{equation}

We assess the finite sample performance of $\widehat{T}_n$ in a Monte Carlo. The results are presented in Table~\ref{tb:mc_results_wj}. We can draw several conclusions from them. First, when volatility is not rough, the statistic $\widehat{T}_n$ performs well, with a relatively small bias only when the starting value of volatility is lowest. In the case of the rough volatility models, the median values of $\widehat{T}_n$ are all strictly positive. This allows for the separation of models with and without rough component of volatility. That said, we do note non-monotonicity in the medians of  $\widehat{T}_n$ as $H$ increases from low to a high value. This is mostly due to the behavior of $\widehat{T}_n$ for the two lowest considered values of $H$. This suggests that for these low values of $H$, higher-order terms in $\widehat{A}_{0,T_1,T_2}(u)$ and $\widehat{A}_{0,T_1,T_3}(u)$, not considered in our expansion, play nontrivial role in the finite sample behavior of these statistics. This is consistent with our expansion results because the size of the approximation error in them increases for lower values of $H$. Intuitively, the differencing in $\widehat{A}_{0,T_1,T_2}(u)$ and $\widehat{A}_{0,T_1,T_3}(u)$ reduces the signal in them for the value of $H$ and makes higher-order terms more important in relative sense. This seems unavoidable if one wants a fully nonparametric jump-robust study of volatility roughness.

\begin{table}[htbp!]
	\caption{Monte Carlo results: Performance of $\widehat{T}_n$ based on $5{,}000$ Monte Carlo replications}
	\begin{tabular*}{0.6\textwidth}{@{\extracolsep{\fill}} cccrrr}
		\hline
		{\bf Case} & & & \multicolumn{3}{c}{Empirical Quantiles of $\widehat{T}_n$}\\
		& & & $25$-th~~ & $50$-th~~ & $75$-th~~\\[1.0ex]
		& & & \multicolumn{3}{c}{Low Volatility}\\[1.0ex]
		$V_0 = 0.015$, $H=0.10$  & & & $0.2298$  & $0.2786$ & $0.3232$\\
		$V_0 = 0.015$, $H=0.25$  & & & $0.5067$  & $0.5568$ & $0.6048$\\
		$V_0 = 0.015$, $H=0.40$  & & & $0.1353$  & $0.1726$ & $0.2098$\\
		$V_0 = 0.015$, $H=0.49$  & & & $0.0203$  & $0.0642$ & $0.1070$\\
		$V_0 = 0.015$, not rough  & & & $-0.0227$  & $0.0350$ & $0.0877$\\[1.0ex] 
		& & & \multicolumn{3}{c}{Medium Volatility}\\[1.0ex]
		$V_0 = 0.030$, $H=0.10$  & & & $0.1281$  & $0.1887$ & $0.2487$\\
		$V_0 = 0.030$, $H=0.25$  & & & $0.3493$  & $0.3958$ & $0.4374$\\
		$V_0 = 0.030$, $H=0.40$  & & & $0.0734$  & $0.1082$ & $0.1433$\\
		$V_0 = 0.030$, $H=0.49$  & & & $-0.0153$  & $0.0357$ & $0.0814$\\
		$V_0 = 0.030$, not rough  & & & $-0.0661$  & $-0.0029$ & $0.0612$\\[1.0ex] 
		& & & \multicolumn{3}{c}{High Volatility}\\[1.0ex]
		$V_0 = 0.060$, $H=0.10$  & & & $0.5054$  & $0.5556$ & $0.6033$\\
		$V_0 = 0.060$, $H=0.25$  & & & $0.1942$  & $0.2256$ & $0.2548$\\
		$V_0 = 0.060$, $H=0.40$  & & & $0.0303$  & $0.0662$ & $0.1041$\\
		$V_0 = 0.060$, $H=0.49$  & & & $-0.0394$  & $0.0169$ & $0.0686$\\
		$V_0 = 0.060$, not rough  & & & $-0.0702$  & $0.0028$ & $0.0759$\\
		\hline
	\end{tabular*}
	\label{tb:mc_results_wj}
\end{table}

\section{Proof of Theorem~\ref{thm:cf} and Proposition~\ref{prop:condmean}}\label{sec:proof:cf}
In a first step, we consider a local expansion of the process $x_{t+T}$ around $T=0$ that only depends on \emph{a priori} estimates. Let $\la_t = f'(v_t)h'(\eta_t)$ and $\wt \la_t = f'(v_t)\wt h'(\wt\eta_t)$ and define
\beq\label{eq:xcheck}
\begin{split}
	\check x_{t+T} &= x_t + \al_t T+\si_t (W_{t+T}- W_t) +\sum_{ \ast\, \in \{\emptyset, \sim\}} f'(v_t)\int_t^{t+T}\int_0^s (\wa g_H(s-r) -\wa g_H(t-r) ) \wa\si^v_{r\wedge t} d\wa W_rdW_s \\
	&\quad+\sum_{ \ast\, \in \{\emptyset, \sim, -, \wedge\}} \la_t\int_t^{t+T}\int_t^s  \int_0^r g_H(s-r)(\wa g^\eta_H(r-w)-\wa g^\eta_H(t-w))\wa\si^\eta_{w\wedge t} d\wa W_w dW_r dW_s \\
	&\quad +\sum_{\ast\, \in \{ \emptyset, \sim, -, \wedge,\circ\}} \wt \la_t  \int_t^{t+T}\int_t^s\int_0^r \wt g_H(s-r)(\wa g^{\wt\eta}_H(r-w)-\wa g^{\wt\eta}_H(t-w)) \wa\si^{\wt\eta}_{w\wedge t} d\wa W_w d\wt W_r dW_s\\
	&\quad+\sum_{\ast\, \in \{\emptyset,\sim, -\}} f'(v_t)\wa\eta^v_t \int_t^{t+T}\int_t^s d\wa W_r dW_s+ \int_t^{t+T} \int_\R\ga(t,z)\mu(ds,dz) \\
	&\quad+\int_t^{t+T} \int_\R\delta(s,z)(\mu-\nu)(ds,dz) +f'(v_t)\int_t^{t+T}\int_t^s \int_\R \delta_v(t,z)(\mu-\nu)(dr,dz)dW_s\\
	&\quad+\frac12\int_t^{t+T}f''(v_t)\biggl(\sum_{ \ast\, \in \{\emptyset, \sim\}} \int_0^s (\wa g_H(s-r) -\wa g_H(t-r) ) \wa\si^v_{r\wedge t} d\wa W_r\biggr)^2dW_s,
\end{split} \raisetag{6\baselineskip}
\eeq
where $\wa x$ simply means $x$ if $\ast = \emptyset$. In what follows, we use the notation $\iint_a^b = \int_a^b\int_\R$ and $A\lec B$ if there is a constant $C\in(0,\infty)$ that does not depend on any important parameter such that $A\leq CB$. Also, we abbreviate $u_T=u/\sqrt{T}$.

\blem\label{lem:xcheck} Under the assumptions of Theorem~\ref{thm:cf}, 
$$ \E_t[e^{iu_T(x_{t+T}-x_t)}]-\E_t[e^{iu_T(\check x_{t+T}-x_t)}] = o^{\mathrm{uc}}(T^{2H}),\quad T\to0.  $$
\elem
\bpr
Since $(v_s-v_t)^2=O(T^{2H})$ uniformly for $s\in[t,t+T]$ and $f\in C^2$, Taylor's theorem implies that we have $x_{t+T}-\check x_{t+T}=\sum_{i=1}^{10} y^{(i)}_T + o(T^{2H})$, where
\begin{align*}
	y_T^{(1)}&=\int_t^{t+T} (\al_s-\al_t) ds,\qquad y_T^{(2)}= f'(v_t) \int_t^{t+T}\int_t^s b_r drdW_s,\\
	y_T^{(3)}&= \sum_{ \ast\, \in \{\emptyset, \sim\}} \wa \la_t \int_t^{t+T}\int_t^s \wa g_H(s-r)(\wa \theta_r-\wa \theta_t) d\wa W_rdW_s,\\
	y_T^{(4)}&=\sum_{ \ast\, \in \{\emptyset, \sim,-,\wedge\}} \la_t \int_t^{t+T}\int_t^s \int_t^r g_H(s-r)\wa g^{\eta}_H(r-w)(\wa \si^\eta_w-\wa \si^\eta_t)d \wa W_w d W_rdW_s, \\
	y_T^{(5)}&=\sum_{ \ast\, \in \{\emptyset, \sim,-,\wedge,\circ\}} \wt \la_t \int_t^{t+T}\int_t^s \int_t^r  \wt g_H(s-r)\wa g^{\wt\eta}_H(r-w)(\wa \si^{\wt\eta}_q-\wa \si^{\wt\eta}_t)d \wa W_w d\wt W_rdW_s,\\
	y_T^{(6)}&= \sum_{ \ast\, \in \{\emptyset, \sim,-\}}f'(v_t)\int_t^{t+T} \int_t^s (\wa \eta^v_r-\wa\eta^v_t) d\wa W_r dW_s,\\
	y_T^{(7)}&=f'(v_t)\int_t^{t+T} \iint_t^s (\delta_v(r,z)-\delta_v (t,z))(\mu-\nu)(dr,dz) dW_s,\\
	y_T^{(8)}&=\frac12 \int_t^{t+T}f''(v_t)\biggl\{(v_s-v_t)^2-\biggl(\sum_{ \ast\, \in \{\emptyset, \sim\}}\int_0^s (\wa g_H(s-r) -\wa g_H(t-r) ) \wa\si^v_{r\wedge t} d\wa W_r\biggr)^2\biggr\}dW_s,\\
	y_T^{(9)}&= f'(v_t)\int_t^{t+T} \iint_t^s \ga_v(r,z)\mu(dr,dz) dW_s,\qquad y_T^{(10)}=\iint_t^T (\ga(s,z)-\ga(t,z))\mu(ds,dz).
\end{align*}
Note that
$
\lvert\E_t[e^{iu_T(x_{t+T}-x_t)}]-\E_t[e^{iu_T(\check x_{t+T}-x_t)}] \rvert\leq \E_t [ uT^{-1/2}\lvert x_{t+T}-\check x_{t+T}\rvert  \wedge 2 ]
$
as a consequence of the elementary inequality $\lvert e^{ix}-e^{ix_0}\rvert\leq \lvert x-x_0\rvert\wedge2$. Because the function $\vp(x)=x\wedge 2$ is subadditive and for all $p\in[0,1]$ there is $C_p>0$ such that  $\vp(x)\leq C_px^p$ for all $x>0$, in order to prove the lemma it suffices to show that 
\beq\label{eq:cases} \begin{cases} \E_t[\lvert y^{(i)}_T\rvert]=o(T^{2H+1/2})
	&\text{for } i=1,\dots, 8,\\
	\E_t[\lvert y^{(i)}_T\rvert^q]=o(T^{2H+q/2}) 
	&\text{for } i=9,10,
\end{cases}\eeq
where $q$ is the parameter from Assumption~\ref{ass:B}.

By \eqref{eq:smooth2}, it is easy to verify that $\E_t[\lvert y^{(1)}_T\rvert] = O(T^{1+H})=o(T^{2H+1/2})$ since $H<\frac12$. For $y^{(2)}_T,\ldots, y^{(6)}_T$, we use Jensen's inequality and bound the $L^2$-norm instead. Using It\^o's isometry in the first step and Minkowski's integral inequality in the second step, we obtain
$$ \E_t[( y^{(2)}_T)^2] =  \int_t^{t+T} \E_t\biggl[\biggl(\int_t^s b_r dr\biggr)^2\biggr] ds  \leq   \int_t^{t+T}  \biggl(\int_t^s \E_t[b_r^2]^{1/2} dr\biggr)^2  ds  =O(T^{3})$$
by \eqref{eq:mom1}, proving that $\E_t[\lvert y^{(2)}_T\rvert] =O(T^{3/2})=o(T^{H+1})=o(T^{2H+1/2})$. Similarly, by \eqref{eq:smooth2},
$$ \E_t[\lvert y^{(6)}_T\rvert] \leq \sum_{ \ast\, \in \{\emptyset, \sim,-\}}\lvert f'(v_t)\rvert\Biggl(\int_t^{t+T} \int_t^s \E_t[(\wa \eta^v_r-\wa\eta^v_t)^2]drds\Biggr)^{1/2}=O(T^{1+H})=o(T^{2H+1/2}).$$
For $y^{(3)}_T$, $y^{(4)}_T$ and $y^{(5)}_T$, note that for $g$ as in Assumption~\ref{ass:A}-5, we have that $\int_t^s g^2(s-r)dr \leq 2(\int_t^s k_H(s-r)^2 dr+\int_t^s \ell_g(s-r)^2 dr)=O((s-t)^{2H})$. Therefore, using \eqref{eq:mom2}  for $y^{(3)}_T$ and \eqref{eq:smooth1}  for $y^{(4)}_T$ and $y^{(5)}_T$, we derive
$ \E_t[ \lvert y^{(3)}_T\rvert] + \E_t[ \lvert y^{(4)}_T\rvert] + \E_t[ \lvert y^{(5)}_T\rvert] = o(T^{1/2+2H}).$
For  $y^{(7)}_T$,  we simply use Jensen's inequality and  \eqref{eq:smooth2} to bound
\begin{equation*}
	\E_t[\lvert y_T^{(7)}\rvert]\leq \lvert f'(v_t)\rvert \Biggl(\int_t^{t+T} \iint_t^s \E_t[(\delta_v(r,z)-\delta_v (t,z))^2] dr \nu(dz)   ds\Biggr)^{1/2}=O(T^{1+H}),
\end{equation*}
which is $o(T^{2H+1/2})$.
For $y_T^{(8)}$, we use the identity $x^2-y^2=2y(x-y)+2(x-y)^2$ together with the bound $v_s-v_t-\sum_{ \ast\, \in \{\emptyset, \sim\}}\int_0^s (\wa g_H(s-r) -\wa g_H(t-r) ) \wa\si^v_{r\wedge t} d\wa W_r=O(T^{2H}\vee T^{1/2})$ to derive the estimate $\E_t[\lvert y_T^{(8)}\rvert]=O(T^{3H+1/2}\vee T^{1+H})=o(T^{2H+1/2})$.

Next, we study $y^{(10)}_T$. Since $(\sum_{i=1}^\infty a_i)^q\leq \sum_{i=1}^r a_i^q$ for any nonnegative numbers $a_i$,  \eqref{eq:smooth3} and \eqref{eq:H}  yield
\begin{align*}
	\E_t[\lvert y^{(10)}_T\rvert^q]\leq \E_t\Biggl[ \iint_t^{t+T} \lvert \ga(s,z)-\ga(t,z)\rvert^q ds\nu(d z)\Biggr]=O(T^{1+qH_\ga})=o(T^{2H+q/2}).
\end{align*}
Finally, using the Burkholder--Davis--Gundy (BDG) inequality in the first step and \eqref{eq:mom4} in the last step, we obtain
\begin{align*}
	&\E_t[\lvert y_T^{(9)}\rvert^q] \leq C_q \lvert f'(v_t)\rvert^q  \E_t\Biggl[\Biggl(\int_t^{t+T} \biggl(\iint_t^s \ga_v(r,z)\mu(dr,dz)\biggr)^2 ds\Biggr)^{q/2} \Biggr]\\
	&\qquad=C_q  \lvert f'(v_t)\rvert^q \E_t\Biggl[\Biggl(\int_t^{t+T} \iint_t^s \iint_t^s \ga_v(r,z)\ga_v(w,q)\mu(dr,dz)\mu(dw,dq) d s\Biggr)^{q/2}\Biggr] \\
	&\qquad=C_q  \lvert f'(v_t)\rvert^q \E_t\Biggl[\Biggl(  \iint_t^{t+T} \iint_t^{t+T} (t+T-r\vee w)\ga_v(r,z)\ga_v(w,q)\mu(dr,dz)\mu(dw,dq) \Biggr)^{q/2}\Biggr] \\
	&\qquad\leq C_q \lvert f'(v_t)\rvert^qT^{q/2} \E_t\Biggl[\Biggl(\iint_t^{t+T} \lvert\ga_v(r,z)\rvert\mu(dr,dz)\Biggr)^q\Biggr]=O(T^{q/2+1})=o(T^{2H+q/2}).\qedhere
\end{align*}
\epr

Next, we define $\delta'$ in the same way as $\delta$ in \eqref{eq:delta}, but with $g_H^\delta(\cdot,z)$ and $\dot g^\delta_H(\cdot,z)$ replaced by $k_{H_\delta}$ from \eqref{eq:kH} for all $z\in\R$. Similarly, we define
$x'_{t+T}$ in the same way as $\check x_{t+T}$ in \eqref{eq:xcheck}, but with all kernels $g$ as in Assumption~\ref{ass:A}-5 replaced by $k_H$ from \eqref{eq:kH} and with $\delta$ replaced by $\delta'$. 
\blem\label{eq:xprime} Under the assumptions of Theorem~\ref{thm:cf}, one has 
$$ \E_t[e^{iu_T(\check x_{t+T}-x_t)}]-\E_t[e^{iu_T (x'_{t+T}-x_t)}] = o^{\mathrm{uc}}(T^{2H}),\quad T\to0.  $$
\elem
\bpr Since $ \lvert\E_t[e^{iu_T(\check x_{t+T}-x_t)}]-\E_t[e^{iu_T( x'_{t+T}-x_t)}] \rvert\leq uT^{-1/2} \E_t [ \lvert \check x_T-  x'_T\rvert ]$, it suffices to show that for each term in \eqref{eq:xcheck} that contains one of the kernels $g$ (or $\delta$), if we replace any of the $g$'s by $\ell_g$ (or $\delta$ by $\delta'$), the resulting term is of order $o(T^{2H+1/2})$. For the substitutions of $g$,  this is straightforward  since by our assumptions on $\ell_g$,
\begin{align*}
	\biggl(\int_t^{t+T}\int_0^s (\ell_g(s-r)-\ell_g(t-r))^2 drds\biggr)^{1/2}&= \biggl(\int_t^{t+T}\int_0^s \biggl(\int_{t\vee r}^s \ell'_g(w-r)dw\biggr)^2 dr ds\biggr)^{1/2} \\
	&\lec\biggl( \int_t^{t+T}\int_0^s ((s-r)^{H+1/2}-(t-r)_+^{H+1/2})^2 dr ds\biggr)^{1/2} \\
	&= O(T^{3/2}),
\end{align*}
which is $o(T^{2H+1/2})$. For the substitution of $\delta$, let us only verify that 
\beq\label{eq:delta-prime} \E_t \biggl[ \biggl\lvert \iint_t^{t+T} \int_0^s (\ell(s-w,z)-\ell(t-w,z))\si^\delta(w,z)dW_w(\mu-\nu)(ds,dz)\biggr\rvert \biggr]=o(T^{2H+1/2})\eeq
for $\ell=\ell_{g^\delta_{H_\delta}}$. (The other expression, involving an integral with respect to $\dot W$, can be analyzed analogously.) We use Jensen's inequality and the BDG inequality in a first step and H\"older's inequality (with respect to $ds$) in a second step to bound the left-hand side of \eqref{eq:delta-prime} by
\begin{align*}
	&C_r	\biggl(\iint_t^{t+T} \E_t\biggl[\biggl\lvert \int_0^s (\ell(s-w,z)-\ell(t-w,z))\si^\delta(w,z) dW_w\biggr\rvert^r\biggr]ds\nu(dz)\biggr)^{\frac1r}\\
	& \leq  C'_r   \biggl(\iint_t^{t+T} \E_t\biggl[ \biggl( \int_0^s (\ell(s-w,z)-\ell(t-w,z))^2\si^\delta(w,z)^2 dw\biggr)^{\frac r2}\biggr] ds\nu(dz)\biggr)^{\frac1r}\\
	& \leq  C'_r  T^{\frac1r-\frac12}  \biggl(\int_\R  \biggl( \int_0^{t+T}\biggl( \int_{w\vee t}^{t+T}(\ell(s-w,z) -\ell(t-w,z))^2 ds \biggr)\E_t[\si^\delta(w,z)^2] dw\biggr)^{\frac r2} \nu(dz)\biggr)^{\frac1r}. 
\end{align*}
By \eqref{eq:mom3}, the last display is $O(T^{1+1/r})=o(T^{2H+1/2})$.
\epr

The next two approximation results  are more subtle, as they do not follow from \emph{a priori} estimates. Here it is crucial to exploit the independence properties of $\mu$ and the different Brownian motions that appear in \eqref{eq:x} and \eqref{eq:sigma}--\eqref{eq:delta}. Let
\beq\label{eq:xprimeprime}
x''_{t+T} = x'_{t+T}+\iint_t^{t+T} (\delta(t,z)-\delta'(s,z))(\mu-\nu)(ds,dz) 
\eeq
and
\begin{equation}\label{eq:xbar}\begin{split}
		\ov x_{t+T}  &= x_t + \al_t T+\si_t (W_{t+T}-W_t)+ \iint_t^{t+T} \ga(t,z)\mu(ds,dz)+\iint_t^{t+T} \delta(t,z)(\mu-\nu)(ds,dz)\\
		&    +f'(v_t)\int_t^{t+T}\iint_t^s  \delta_v(t,z)(\mu-\nu)(dr,dz)dW_s + f'(v_t)\eta^v_t \int_t^{t+T}\int_t^s d  W_r dW_s\\
		&  +\sum_{ \ast\, \in \{\emptyset, \sim\}} f'(v_t) \int_t^{t+T}\int_0^s (k_H(s-r)-k_H(t-r))\wa\si^v_{r\wedge t} d\wa W_rdW_s \\
		& +\sum_{\ast\,\in\{\emptyset,\sim,-, \wedge\}} \la_t \int_t^{t+T} \int_t^s k_H(s-r) \int_0^t (k_H(r-w)-k_H(t-w)) \wa\si^\eta_wd\wa W_w dW_r dW_s\\
		& +  \la_t\si^\eta_t\int_t^{t+T}\int_t^s \int_t^r k_H(s-r)  k_H(r-w) d  W_w dW_r dW_s \\
		& +\frac12f''(v_t) \int_t^{t+T}\int_0^s(k_H(s-r)-k_H(t-r))^2[(  \si^v_{r\wedge t})^2+(\wt \si^v_{r\wedge t})^2 ] drdW_s\\
		& +f''(v_t)\int_t^{t+T}\int_0^s\int_0^r(k_H(s-r)-k_H(t-r))(k_H(s-w)-k_H(t-w))\\
		&\qquad\qquad\qquad\qquad\times(\si^v_{w\wedge t}dW_w+\wt\si^v_{w}\bone_{[0,t]}(w) d\wt W_w) (\si^v_{r\wedge t}dW_r+\wt\si^v_{r}\bone_{[0,t]}(r) d\wt W_r) dW_s.\end{split}\raisetag{6\baselineskip}
\end{equation} 
\blem\label{lem:xprimeprime} 
Under the assumptions of Theorem~\ref{thm:cf}, we have 
\begin{align*}
	&\E_t[e^{iu_T  (x'_{t+T}-x_t)}]-\E_t[e^{iu_T  (x''_{t+T}-x_t)}] \\	
	&\qquad = e^{-\frac12 u^2\si_t^2+T\varphi_t(u_T)} \biggl(\sum_{i=1}^3 C_{1,i}(u,T)_t +C'_{1,1}(u,T)_t \biggr)+ o^{\mathrm{uc}}(T^{2H}),\qquad T\to0.
\end{align*} 
\elem
\blem\label{lem:xbar} 
Under the assumptions of Theorem~\ref{thm:cf}, we have  
$$ \E_t[e^{iu_T  (x''_{t+T}-x_t)}]-\E_t[e^{iu_T  (\ov x_{t+T}-x_t)}] = o^{\mathrm{uc}}(T^{2H}),\quad T\to0. $$
\elem
\bpr[Proof of Lemma~\ref{lem:xprimeprime}] For any $r\in[1,2)$, there is a constant $C_r\in(0,\infty)$ such that  $\lvert e^{ix}-e^{ix_0}-ie^{ix_0}(x-x_0)\rvert \leq C_r\lvert x-x_0\rvert^r$ for all $x_0,x\in\R$. Thus,
\beq\label{eq:aux}\begin{split} &\Bigl\lvert \E_t[e^{iu_T  (x'_{t+T}-x_t)}]-\E_t[e^{iu_T  (x''_{t+T}-x_t)}] -iu_T\E_t[ e^{iu_T(x''_{t+T}-x_t)} (x'_{t+T}-x''_{t+T})] \Bigr\rvert\\
	&\qquad\leq C_r\lvert u\rvert ^rT^{-r/2}\E_t[\lvert x'_{t+T}-x''_{t+T}\rvert^r], \end{split}\eeq
By the BDG inequality and \eqref{eq:smooth4}, 
$$
\lvert u\rvert^rT^{-r/2}\E_t[\lvert x'_{t+T}-x''_{t+T}\rvert^r] 
\leq C'_r\lvert u\rvert^rT^{-r/2}\E_t\biggl[\iint_t^{t+T} \lvert\delta'(s,z)-\delta(t,z)\rvert^r \nu(dz)ds\biggr]
$$ 
(the difference between $\delta$ and $\delta'$ can be neglected, as seen before).
It is straightforward to show that \eqref{eq:delta} (together with Assumptions~\ref{ass:A} and \ref{ass:B}) implies  
\beq\label{eq:smooth-delta}  \E_t\biggl[\int_\R \lvert\delta'(s,z)-\delta(t,z)\rvert^r \nu(dz)\biggr] = O((s-t)^{rH_\delta}).\eeq
Therefore, \eqref{eq:aux} is of order $O^{\mathrm{uc}}(T^{1-r/2+rH_\delta})$,
which is $o^{\mathrm{uc}}(T^{2H})$   by \eqref{eq:H}.

It thus remains to evaluate $u_T\E_t[ e^{iu_T(x''_{t+T}-x_t)} (x'_{t+T}-x''_{t+T})]$. We claim that
\beq\label{eq:claim}
\E_t \biggl[ \biggl\lvert x'_{t+T}-x''_{t+T}-\sum_{i=1}^5 x^{(i)}_{t+T}  \biggr\rvert^r\biggr]^{1/r} = o(T^{1/2+2H}),
\eeq
where
\begin{equation}\label{eq:xTi}\begin{split}
		x^{(1)}_{t+T} &= \iint_t^{t+T} \int_0^s (k_{H_\delta}(s-w)-k_{H_\delta}(t-w)) \si^\delta(w\wedge t,z)d W_w (\mu-\nu)(ds,dz),\\
		x^{(2)}_{t+T} &=\iint_t^{t+T}  \int_0^s (k_{H_\delta}(s-w)-k_{H_\delta}(t-w)) \dot \si^\delta(w \wedge t,z)d \dot W_w (\mu-\nu)(ds,dz),\\
		x^{(3)}_{t+T} &=\iint_t^{t+T} \int_t^s  \eta^\delta(t,z)d W_w (\mu-\nu)(ds,dz),\\
		x^{(4)}_{t+T} &= \iint_t^{t+T} \int_t^s \ddot\eta^\delta(t,z)d \ddot W_w (\mu-\nu)(ds,dz),\\
		x^{(5)}_{t+T} &= \iint_t^{t+T} \iint_t^s \delta_\delta(t,z,z') (\mu-\nu)(dw,dz')(\mu-\nu)(ds,dz).
\end{split}\end{equation}
To prove this claim,
we only consider one particular contribution (the largest one) to the difference in \eqref{eq:claim}, namely
$$	\E_t \biggl[ \biggl\lvert \iint_t^{t+T} \int_t^s k_{H_\delta}(s-w)(\si^\delta(w,z)-\si^\delta(t,z)) dW_w (\mu-\nu)(ds,dz) \biggr\rvert^r\biggr]^{1/r},$$
and leave the remaining terms (which can be treated analogously) to the reader. Similarly to how we estimated \eqref{eq:delta-prime}, we can bound the term in the previous display by a constant times
\begin{align*}
	T^{1/r-1/2}\biggl(  \int_\R  \biggl( \int_t^{t+T}\biggl( \int_{w\vee t}^{t+T}k_{H_\delta}(s-w)^2 ds \biggr)\E_t[(\si^\delta(w,z)-\si^\delta(t,z))^2] dw\biggr)^{r/2} \nu(dz)\biggr)^{1/r}, 
\end{align*}
which is  $O(T^{1/r+2H_\delta})=o(T^{1/2+2H})$ by \eqref{eq:smooth4}, \eqref{eq:H} and the fact that $ \int_w^{t+T} k_{H_\delta}(s-w)^2 ds = O(T^{2H_\delta})$.  

As a consequence, only $u_T\E_t[ e^{iu_T(x''_{t+T}-x_t)} \sum_{i=1}^5 x^{(i)}_{t+T}]$ needs to be considered further.
For all $i=1,\dots, 5$, we have $ T^{-1/2}\E_t[\lvert x^{(i)}_{t+T}\rvert^r]^{1/r} = O(T^{1/r-1/2+H_\delta})=o(T^{2H/r})$   by  \eqref{eq:H}. If $r=1$, this is $o(T^{2H})$, so we can assume $r\in(1,2)$ in the following. Let $r^\ast=1/(1-1/r)$ be the conjugate exponent of $r$. Using H\"older's inequality and the estimate $\E[\lvert e^{ix}-e^{iy}\rvert^{r_\ast}]^{1/r_\ast}\leq C\E[( \lvert x-y\rvert^{r_\ast} \wedge 1]^{1/r_\ast}\leq C(\E[\lvert x-y\rvert^2]^{1/r_\ast}\wedge \E[\lvert x-y\rvert]^{1/r_\ast})$,  
we deduce that 
\begin{align*} 
	u_T \E_t\biggl[ e^{iu_T(x''_{t+T}-x_t)} \sum_{i=1}^5 x^{(i)}_{t+T}\biggr]&= u_T\E_t \biggl[ 
	e^{iu_T\si_t(W_{t+T}-W_t) + iu_T\iint_t^{t+T}\delta(t,z)(\mu-\nu)(ds,dz)}\\
	&\quad\times e^{ iu_T\iint_t^{t+T}\ga(t,z)\mu(ds,dz)} \sum_{i=1}^5 x^{(i)}_{t+T}\biggr]  + o^{\mathrm{uc}}(T^{2H}).
\end{align*}

Therefore,
upon defining $M_{\tau}=\si_t(W_{\tau}-W_t)+\iint_t^{\tau}\delta(t,z)(\mu-\nu)(ds,dz)+\iint_t^\tau \ga(t,z)\mu(ds,dz)$ for $\tau\geq t$, we have that
\begin{equation}\label{eq:5terms}u_T \E_t[ e^{iu_T(x''_{t+T}-x_t)} (x'_{t+T}-x''_{t+T})]  = u_T \E_t\biggl[ e^{iu_TM_{t+T}} \sum_{i=1}^5 x^{(i)}_{t+T}\biggr] + o^{\mathrm{uc}}(T^{2H}).\end{equation}
By It\^o's lemma (see \cite[Chapter~I, Theorem~4.57]{JS03}), the process $V(u)_\tau=e^{iuM_\tau}$ satisfies the stochastic differential equation (SDE)
\beq\label{eq:SDE}\begin{split} dV(u)_\tau&= iu\si_tV(u)_\tau  dW_\tau + \int_\R V(u)_{\tau-} (e^{iu(\delta(t,z)+\ga(t,z))}-1)(\mu-\nu)(d\tau,dz)\\
	&\quad + b(u)V(u)_\tau d\tau,\quad \tau\geq t,\\ V(u)_t&=1,  \end{split}\eeq
where 
\beq\label{eq:b} b(u)=-\frac12 u^2\si_t^2 + \int_\R (e^{iu(\delta(t,z)+\ga(t,z))}-1-iu\delta(t,z))\nu(dz).\eeq

We are now in the position to compute $\E_t [ e^{iu_TM_{t+T}} \sum_{i=1}^5 x^{(i)}_{t+T} ]$ for each $i=1,\dots,5$. In a first step, note that $V(u)_\tau$ is measurable with respect to the $\si$-algebra $\calg$ generated by $\calf_t$, $W$ and $\mu$. Since we can rewrite
\begin{equation*} 
	x^{(4)}_{t+T}= \int_t^{t+T} \iint_{w}^{t+T}  \ddot\eta^\delta(t,z)(\mu-\nu)(ds,dz)d\ddot W_w,
\end{equation*}
and 
$d \ddot W_\tau$ for $\tau\geq t$ is independent of $\calg$, it follows that 
\beq\label{eq:term24}\E_t [ e^{iu_TM_{t+T}}  x^{(4)}_{t+T} ] = \E_t[e^{iu_TM_{t+T}} \E[x^{(4)}_{t+T}\mid \calg]]=0.\eeq

Next, we consider $i=1$. Using integration by parts, we have that
\begin{align*}
	\E_t [ e^{iu_TM_{t+T}}  x^{(1)}_{t+T} ] &= \iint_t^{t+T} \E_t\biggr[V(u_T)_s\int_0^s (k_{H_\delta}(s-w)-k_{H_\delta}(t-w))\si^\delta(w\wedge t,z)dW_w\biggr] \\
	& \quad\times (e^{iu_T(\delta(t,z)+\ga(t,z))}-1)\nu(dz)ds+b(u_T)\int_t^{t+T} \E_t[e^{iu_T M_s}x_s^{(1)}] ds.
\end{align*}
So if we denote $m(u,z)_{r,s}=	\E_t [V(u)_r\int_0^r (k_{H_\delta}(s-w)-k_{H_\delta}(t-w))\si^\delta(w\wedge t,z)dW_w ]$ for $t\leq r\leq s$, then  
\beq\label{eq:aux2}  \E_t [ e^{iu_TM_{t+T}}  x^{(1)}_{t+T} ] =\iint_t^{t+T} e^{b(u_T)(t+T-s)}m(u_T,z)_{s,s}(e^{iu_T(\delta(t,z)+\ga(t,z))}-1)\nu(dz)ds.  \eeq
Using integration by parts one more time, we derive the identity
\begin{align*}
	m(u_T,z)_{r,s}&=\E_t[V(u_T)_t]\int_0^t (k_{H_\delta}(s-w)-k_{H_\delta}(t-w))\si^\delta(w,z)dW_w \\
	&\quad+iu_T\si_t  \si^\delta(t,z)\int_t^r k_{H_\delta}(s-w) \E_t [V(u_T)_w ] dw + b(u_T)  \int_t^r m(u_T,z)_{w,s} dw.
\end{align*}
Since $\E_t [V(u_T)_w ] =e^{b(u_T)(w-t)}$, it follows that
\begin{align*} 	m(u_T,z)_{r,s}&= e^{b(u_T)(r-t)}\int_0^t (k_{H_\delta}(s-w)-k_{H_\delta}(t-w))\si^\delta(w,z)dW_w \\
	&\quad+ iu_T\si_t\si^\delta(t,z)e^{b(u_T)(r-t)}\int_t^r  k_{H_\delta}(s-w)  dw \\
	&= e^{b(u_T)(r-t)}\int_0^t (k_{H_\delta}(s-w)-k_{H_\delta}(t-w))\si^\delta(w,z)dW_w\\
	&\quad+\frac{iu_T\si_t\si^\delta(t,z) e^{b(u_T)(r-t)}}{\Ga(H_\delta+\frac32)}[(s-t)^{H_\delta+1/2}-(s-r)^{H_\delta+1/2}]. \end{align*}
Inserting this with $r=s$ in \eqref{eq:aux2} yields
\begin{align*}
	\E_t [ e^{iu_TM_{t+T}}  x^{(1)}_{t+T} ] &= e^{b(u_T)T}\int_t^{t+T}   \int_0^t [k_{H_\delta}(s-w)-k_{H_\delta}(t-w)]\\
	&\qquad\times \int_\R \si^\delta(w,z) (e^{iu_T(\delta(t,z)+\ga(t,z)}-1)\nu(dz) dW_w ds \\
	&\quad+\frac{iu_T\si_t e^{b(u_T)T}}{\Ga(H_\delta+\frac52)}  T^{H_\delta+3/2} \int_\R \si^\delta(t,z)(e^{iu_T(\delta(t,z)+\ga(t,z))}-1)\nu(dz).
\end{align*}
Note that $e^{b(u_T)T}=e^{-\frac12\si_t^2 u^2+T\vp_t(u_T)}+O(T^{1/2})$ and $T\int_\R \si^\delta(t,z)(e^{iu_T(\delta(t,z)+\ga(t,z))}-1)\nu(dz)=T\int_\R \si^\delta(t,z)(e^{iu_T\delta(t,z)}-1)\nu(dz)+O(T^{1/2})$. 
Thus, in the notation of \eqref{eq:phipsi} and \eqref{eq:chi}, the contribution of $x^{(1)}_{t+T}$ to \eqref{eq:5terms} is 
\beq\label{eq:term1} \begin{split} u_T\E [ e^{iu_TM_{t+T}}  x^{(1)}_{t+T} ]  &= ue^{-\frac12\si_t^2 u^2+T\vp_t(u_T)}T^{1/2}\int_0^1 \int_0^t[k_{H_\delta}(t+sT-w)-k_{H_\delta}(t-w)] \\
	&\qquad\times \int_\R \si^\delta(w,z) (e^{iu_T(\delta(t,z)+\ga(t,z))}-1)\nu(dz)  dW_w ds \\
	&\quad+\frac{iu^2\si_te^{-\frac12\si_t^2 u^2+T\vp_t(u_T)}}{\Ga(H_\delta+\frac52)} T^{H_\delta+1/2}\chi^{(1)}_t(u_T)+o^{\mathrm{uc}}(T^{2H}). \end{split}\eeq
A similar argument can be employed to analyze $x^{(2)}_{t+T}$, $x^{(3)}_{t+T}$ and $x^{(5)}_{t+T}$. As a result,
\begin{align}
	\nonumber u_T\E [ e^{iu_TM_{t+T}}  x^{(2)}_{t+T} ] &=ue^{-\frac12\si_t^2 u^2+T\vp_t(u_{T})}T^{1/2}\int_0^1 \int_0^t[k_{H_\delta}(t+sT-w)-k_{H_\delta}(t-w)] \\
	&\quad\times \int_\R \dot\si^\delta(w,z) (e^{iu_T(\delta(t,z)+\ga(t,z))}-1)\nu(dz)  d\dot W_w ds	\label{eq:term2} \end{align}
and 
\begin{align}
	\label{eq:term3} u_TE [ e^{iu_TM_{t+T}}  x^{(3)}_{t+T} ] &= \frac12iu^2\si_te^{-\frac12\si_t^2 u^2+T\vp_t(u_{T})}  T\chi^{(2)}_t(u_T),\\
	\label{eq:term5}u_T\E [ e^{iu_TM_{t+T}}  x^{(5)}_{t+T} ] &=\frac12ue^{-\frac12\si_t^2 u^2+T\vp_t(u_{T})}T^{3/2}\chi^{(3)}_t(u_T).\end{align}
Recalling \eqref{eq:C1i}, we deduce the lemma  by combining \eqref{eq:5terms}, \eqref{eq:term24} and \eqref{eq:term1}--\eqref{eq:term5}.
\epr

\bpr[Proof of Lemma~\ref{lem:xbar}] 
By first conditioning on $W$ and $\wt W$, we obtain
$$
\E_t\Bigl[e^{iu_T(x''_{t+T}-x_t)}\mathrel{\big|} W,\wt W\Bigr] =\prod_{i=1}^5 A^{(i)}_T,
$$
where
\begin{align*}
	A^{(1)}_T&=\exp\biggl(iu_T\biggl(\al_tT+\si_t(W_{t+T}-W_t) + \sum_{ \ast\, \in \{\emptyset, \sim\}}f'(v_t)\wa\eta^v_t \int_t^{t+T}\int_t^s d\wa W_r dW_s\\
	&\quad+ \sum_{ \ast\, \in \{\emptyset, \sim\}} f'(v_t) \int_t^{t+T}\int_0^s (k_H(s-r)-k_H(t-r))\wa \si^v_{r\wedge t}d\wa W_rdW_s\\
	& \quad+ \sum_{ \ast\, \in \{\emptyset, \sim\}}\la_t\int_t^{t+T}\int_t^s \int_0^r k_H(s-r)(k_H(r-w)-k_H(t-w)) \wa\si^\eta_{w\wedge t}d\wa W_w dW_r dW_s\\
	& \quad+ \sum_{ \ast\, \in \{\emptyset, \sim\}}\wt \la_t\int_t^{t+T}\int_t^s \int_0^r k_H(s-r)(k_H(r-w)-k_H(t-w))\wa\si^{\wt\eta}_{w\wedge t} d\wa W_w d\wt W_r dW_s\\
	&\quad + \sum_{\ast\, \in \{ -,\wedge\}} \la_t\int_t^{t+T}\int_t^s \int_0^t k_H(s-r)(k_H(r-w)-k_H(t-w)) \wa\si^\eta_{w}d\wa W_w dW_r dW_s\\
	&\quad+ \sum_{\ast\, \in \{ -,\wedge,\circ\}} \wt \la_t \int_t^{t+T}\int_t^s \int_0^t k_H(s-r)(k_H(r-w)-k_H(t-w))\wa\si^{\wt\eta}_{w} d\wa W_w d\wt W_r dW_s\\
	&\quad+\frac12\int_t^{t+T}f''(v_t) \int_0^s(k_H(s-r)-k_H(t-r))^2[(  \si^v_{r\wedge t})^2+(\wt \si^v_{r\wedge t})^2 ] drdW_s\\
	&\quad+\int_t^{t+T}f''(v_t)\int_0^s\int_0^r(k_H(s-r)-k_H(t-r))(k_H(s-w)-k_H(t-w))\\
	&\qquad\qquad\qquad\qquad\qquad\qquad\times(\si^v_{w\wedge t}dW_w+\wt\si^v_{w\wedge t} d\wt W_w) (\si^v_{r\wedge t}dW_r+\wt\si^v_{r\wedge t} d\wt W_r) dW_s\biggr)\biggr),\\
	A^{(2)}_T&=\exp\biggl(\iint_t^{t+T} \Bigl(e^{iu_T(\delta(t,z)+\ga(t,z)+f'(v_t)\delta_v(t,z)(W_{t+T}-W_s))}-1\\
	&\qquad\qquad\qquad\qquad\qquad\qquad\quad -iu_T(\delta(t,z)+f'(v_t)\delta_v(t,z)(W_{t+T}-W_s))\Bigr)\nu(dz)ds\biggr),\\
	A^{(3)}_T&=\exp\biggl(-\frac12u^2_T\int_t^{t+T} \biggl( f'(v_t)\ov\eta^v_t(W_{t+T}-W_w)+\la_t\ov\si^\eta_t\int_w^{t+T}\int_w^s k_H(s-r)k_H(r-w)dW_rdW_s\\
	&\quad +\wt \la_t\ov\si^{\wt\eta}_t\int_w^{t+T}\int_w^s k_H(s-r)k_H(r-w)d\wt W_rdW_s  \biggr)^2 dw\biggr),\\
	A^{(4)}_T&= \exp\biggl(-\frac12u^2_T\int_t^{t+T} \biggl(\la_t\wh \si^\eta_t\int_w^{t+T}\int_w^s k_H(s-r)k_H(r-w)dW_rdW_s\\
	&\quad+\wt \la_t\wh\si^{\wt\eta}_t\int_w^{t+T}\int_w^s k_H(s-r)k_H(r-w)d\wt W_rdW_s \biggr)^2 dw\biggr),\\
	A^{(5)}_T&= \exp\biggl(-\frac12u^2_T\int_t^{t+T} \biggl(\wt \la_t\mathring \si^{\wt\eta}_t\int_w^{t+T}\int_w^s k_H(s-r)k_H(r-w)d\wt W_rdW_s \biggr)^2dw\biggr).
\end{align*}
Since $\lvert A^{(i)}_T\rvert\leq 1$ for all $i$ and $\lvert \prod_{i=1}^n a_i-\prod_{i=1}^n b_i\rvert \leq \sum_{i=1}^n \lvert a_i-b_i\rvert$ for any complex numbers $a_1,\dots,a_n$ and $b_1,\dots,b_n$ in the unit disk, we can replace any $A^{(i)}_T$ by $1$ if $\E_t[\lvert A^{(i)}_T-1\rvert] =o^\mathrm{uc}(T^{2H})$. We shall apply this to $i=3,4,5$. Indeed, using the bound $\rvert 1-e^{-x}\rvert\leq x$, we have 
$
\E_t[\lvert A^{(3)}_T-1\rvert]=O^{\mathrm{uc}}(T\vee T^{4H})$, $\E_t[\lvert A^{(4)}_T-1\rvert]=O^{\mathrm{uc}}(T^{4H})$ and  $\E_t[\lvert A^{(5)}_T-1\rvert]=O^{\mathrm{uc}}(T^{4H})$.
This shows that
\beq\label{eq:intermed} \E_t[e^{iu_T (x''_{t+T}-x_t)}]=\E_t[A^{(1)}_TA^{(2)}_T]+o^\mathrm{uc}(T^{2H}). \eeq

Next, we take expectation with respect to $\wt W$, still conditioning on $W$, which amounts to computing $\E_t[A^{(1)}_T\mid W]$. 
Conditionally on $W$ and $\calf_t$, $A_T^{(1)}$ belongs to the direct sum of Wiener chaoses (with respect to $(\wt W_\tau-\wt W_t)_{\tau\geq t}$) up to order $2$. Hence, by Theorem~\ref{thm:P2}, 
$$
\E_t[A^{(1)}_T\mid W] = A^{(11)}_T A^{(12)}_T 
$$
where 
\begin{align*}
	A^{(11)}_T&=\exp\biggl(iu_T\biggl(\al_t{T}+\si_t(W_{t+T}-W_t) + f'(v_t) \eta^v_t  \int_t^{t+T}\int_t^s d  W_r dW_s \\
	& \quad+f'(v_t)  \int_t^{t+T}\int_0^s (k_H(s-r)-k_H(t-r))\si^v_{r\wedge t}d  W_rdW_s\\
	& \quad+  f'(v_t)\int_t^{t+T}\int_0^t (k_H(s-r)-k_H(t-r))\wt \si^v_{r}d \wt W_rdW_s\\
	& \quad+ \la_t \int_t^{t+T}\int_t^s \int_0^r k_H(s-r)(k_H(r-w)-k_H(t-w))\si^\eta_{w\wedge t} d  W_wdW_r dW_s\\
	&\quad+ \la_t\sum_{\ast\, \in \{ \sim,-,\wedge\}}  \int_t^{t+T}\int_t^s \int_0^t k_H(s-r)(k_H(r-w)-k_H(t-w)) \wa\si^\eta_{w}d\wa W_wdW_r dW_s\\
	&\quad+\frac12f''(v_t) \int_t^{t+T}\int_0^s(k_H(s-r)-k_H(t-r))^2[(  \si^v_{r\wedge t})^2+(\wt \si^v_{r\wedge t})^2 ] drdW_s\\
	&\quad+f''(v_t)\int_t^{t+T}\int_0^s\int_0^r(k_H(s-r)-k_H(t-r))(k_H(s-w)-k_H(t-w))\\
	&\qquad\qquad\qquad \times(\si^v_{w\wedge t}dW_w+\wt\si^v_{w}\bone_{[0,t]}(w) d\wt W_w) (\si^v_{r\wedge t}dW_r+\wt\si^v_{r}\bone_{[0,t]}(r) d\wt W_r) dW_s\biggr)\biggr),\\
	A^{(12)}_T	&=\exp\biggl(-\frac12\sum_{j=1}^\infty \Bigl[\log(1-2i\al_ju_T)+2i\al_j u_T+\frac{\beta_j^2u^2_T}{1-2i\al_ju_T}\Bigr]\biggr)
\end{align*}
and $(\al_j)_{j\geq1}$ and $(\beta_j)_{j\geq1}$ are $(\calf_t\vee \si(W))$-measurable random variables. If $Z^{(1)}_T$ and $Z^{(2)}_T$ denote the projections of $A^{(1)}_T$ onto the first- and second-order Wiener chaos generated by $(\wt W_{t+\tau})_{\tau\geq0}$, respectively, then 
\begin{align*}
	\sum_{j=1}^\infty \al_j^2 &= \frac12\E_t[(Z^{(2)}_T)^2\mid W] = \frac12\int_t^{t+T}\int_t^r \biggl(\wt \la_t\wt\si^{\wt \eta}_tk_H(r-w)\int_r^{t+T} k_H(s-r) dW_s\\
	&\quad+f''(v_t)(\wt \si^v_t)^2\int_r^{t+T} k_H(s-r)k_H(s-w)dW_s\biggr)^2dw dr,\\
	\sum_{j=1}^\infty \beta_j^2 &=\E_t[(Z^{(1)}_T)^2\mid W]=\int_t^{t+T}\biggl(f'(v_t)\wt\eta^v_t (W_{t+T}-W_w)+f'(v_t)\wt \si^v_t\int_w^{t+T} k_H(s-w)dW_s\\
	&\quad+ \la_t\wt\si^\eta_t\int_w^{t+T}\int_w^s   k_H(s-r)k_H(r-w)   dW_r dW_s\\
	&\quad+ \wt \la_t\int_w^{t+T}\int_0^w k_H(s-w)(k_H(w-r)-k_H(t-r))\si^{\wt \eta}_{r\wedge t} dW_r dW_s\\
	&\quad+ \wt \la_t\sum_{\ast\, \in \{\sim, -,\wedge,\circ\}}  \int_w^{t+T} \int_0^t k_H(s-w)(k_H(w-r)-k_H(t-r))\wa\si^{\wt\eta}_{r} d\wa W_r  dW_s\\
	&\quad+f''(v_t)\wt\si^v_t \int_w^{t+T}\int_0^s (k_H(s-r)-k_H(t-r))(k_H(s-w)-k_H(t-w))\\
	&\qquad\qquad\qquad\qquad\qquad\qquad\qquad\qquad\times(\si^v_{r\wedge t}  dW_r+\wt \si^v_r\bone_{[0,t]}(r)d\wt W_r )dW_s  \biggr)^2 dw.
\end{align*}
Since $\log(1+x)=x+O(x^2)$ and $(1-x)^{-1}=1+O(x)$, 
$$ A^{(12)}_T=\exp\biggl(-\frac12u^2_T\sum_{j=1}^\infty \beta_j^2+O^{\mathrm{uc}}\biggl(T^{-1}\sum_{j=1}^\infty \al_j^2\biggr)+O^{\mathrm{uc}}\biggl(T^{-3/2}\sum_{j=1}^\infty \al_j\beta_j^2\biggr)\biggr). $$
A straightforward calculation shows that 
$ \E_t [\sum_{j=1}^\infty \al_j^2 ] = O^{\mathrm{uc}}(T^{1+4H})$ and 
\begin{align*}
	\E_t\biggl[\sum_{j=1}^\infty \al_j\beta_j^2\biggr] &\leq \E_t\biggl[\biggl(\max_{j\geq1}\al_j\biggr)\sum_{j=1}^\infty \beta_j^2\biggr]\leq \E_t\biggl[ \biggl(\sum_{j=1}^\infty \al_j^2\biggr)^{1/2}\sum_{j=1}^\infty \beta_j^2\biggr]\\
	&\leq \E_t\biggl[  \sum_{j=1}^\infty \al_j^2\biggr]^{1/2}\E_t\biggl[ \biggl( \sum_{j=1}^\infty \beta_j^2\biggr)^2\biggr]^{1/2}=O^{\mathrm{uc}}(T^{3/2+4H}),
\end{align*}
which, combined with \eqref{eq:intermed} and the fact that 
$$\E_t\Biggl[\Biggl\lvert\sum_{j=1}^\infty \beta_j^2-(f'(v_t)\wt \si^v_t)^2  \int_t^{t+T} \biggl(\int_v^{t+T} k_H(s-w)dW_s\biggr)^2 dw\Biggr\rvert\Biggr]=O^{\mathrm{uc}}(T(T^{H+1/2}\vee T^{3H})),$$
proves that 
\beq\label{eq:intermed2}  \E_t[e^{iu_T  (x''_{t+T}-x_t)}]=\E_t[A^{(11)}_TA^{(2)}_Te^{-\frac12 u^2_T(f'(v_t)\wt \si^v_t)^2\int_t^{t+T}  (\int_w^{t+T} k_H(s-w)dW_s )^2 dw}]+o^\mathrm{uc}(T^{2H}). \eeq
The last expectation is exactly $\E_t[e^{iu_T  (\ov x_{t+T}-x_t)}]$, as the reader can easily confirm.
\epr

Next, let us define
\beq\label{eq:X} X_T=X^{(0)}_T+X^{(1)}_T+X^{(2)}_T,\qquad X'_T=X_T+X^{(3)}_T,\eeq
where $X^{(0)}_T=\al_t\sqrt{T}$ and
\begin{align*}
	X^{(1)}_T&= T^{-1/2} \int_t^{t+T} \biggl(\si_t+\sum_{\ast\,\in\{\emptyset,\sim\}} f'(v_t)\int_0^t (k_H(s-r)-k_H(t-r))\wa\si^v_r d\wa W_r\\
	&\qquad\qquad+\frac12f''(v_t)\int_0^s (k_H(s-r)-k_H(t-r))^2[(\si^v_{r\wedge t})^2+(\wt \si^v_{r\wedge t})^2] dr\\
	&\qquad\qquad+f''(v_t)\int_0^t\int_0^t(k_H(s-r)-k_H(t-r))(k_H(s-w)-k_H(t-w))\\
	&\qquad\qquad\qquad\qquad\qquad\qquad\qquad\times(\si^v_wdW_w+\wt\si^v_{w}  d\wt W_w) (\si^v_rdW_r+\wt\si^v_{r}  d\wt W_r)\biggr)dW_s\\
	& = T^{-1/2}\int_t^{t+T} \si_{s\mid t} dW_s + o(T^{2H}), \\
	X^{(2)}_T&= T^{-1/2}\int_t^{t+T}\int_t^s \biggl(f'(v_t)\eta^v_t +f'(v_t)\si^v_t k_H(s-r)\\
	&\quad+\sum_{\ast\,\in\{\emptyset,\sim,-,\wedge\}} \la_tk_H(s-r)\int_0^t (k_H(r-w)-k_H(t-w))\wa\si^\eta_w d\wa W_w\\
	&\quad+f''(v_t)\si^v_tk_H(s-r)\int_0^t (k_H(s-w)-k_H(t-w)) (\si^v_{w}dW_w+\wt \si^v_wd\wt W_w)\biggr) d  W_r dW_s  \\
	& =T^{-1/2}\int_t^{t+T}\int_t^s \Bigl (f'(v_t)\eta^v_t +   f'(v_t)\si^v_{r\mid t}k_H(s-r) \\
	&\qquad\qquad\qquad\qquad\qquad\qquad\quad +f''(v_t)\si^v_tk_H(s-r)(v_{s\mid t}-v_t)\Bigr) d  W_r dW_s + o(T^{2H}) ,\\
	X^{(3)}_T&=T^{-1/2}   \int_t^{t+T}\int_t^s \int_t^r \Bigl(\la_t\si^\eta_t k_H(s-r)k_H(r-w)\\
	&\qquad\qquad\qquad\qquad\qquad\qquad \qquad  + f''(v_t)(\si^v_t)^2k_H(s-r)k_H(s-w)\Bigr) d  W_w dW_r dW_s.
\end{align*}
Furthermore, define
\beq\label{eq:Y} Y_T=T^{-1}(f'(v_t)\wt \si^v_t)^2\int_t^{t+T} \biggl (\int_w^{t+T} k_H(s-w)dW_s \biggr)^2 dw.\eeq
In particular, $A^{(11)}_T=e^{iu X'_T}$ and  by \eqref{eq:intermed2},
\beq\label{eq:intermed3}  \E_t[e^{iu_T  (\ov x_{t+T}-x_t)}]=\E_t[e^{iuX'_{T}}A^{(2)}_T e^{-\frac12u^2Y_T}]+o^\mathrm{uc}(T^{2H}). \eeq
\blem\label{lem:As}
Under the assumptions of Theorem~\ref{thm:cf}, we have  
\beq\label{eq:cf:xbar}\begin{split}\E_t[e^{iu_T  (\ov x_{t+T}-x_t)}] &= e^{T\psi_t(u_{T})}\E_t[e^{iuX_T}]+e^{-\frac12 u^2\si_t^2+T\varphi_t(u_{T})} C_{1,4}(u,T)_t\\
	&\quad+ \wt C_2(u,T)_t +o^{\mathrm{uc}}(T^{2H})
\end{split}\eeq
as $T\to0$, where
\beq\label{eq:C0}
\wt C_2(u,T)_t=iu  \E_t [e^{iu_T\si_t(W_{t+T}-W_t)} X_T^{(3)}  ] - \frac12u^2\E_t[e^{iu_T\si_t(W_{t+T}-W_t)}Y_T].
\eeq
\elem
\bpr
Since $X^{(3)}_T=O(T^{2H})$, $Y_T=O(T^{2H})$, $X_T=\si_tT^{-1/2}(W_{t+T}-W_t)+O(T^H)$ and $A^{(2)}_T=1+o^{\mathrm{uc}}(1)$ (we will show the latter at the end of this proof), Taylor's theorem and \eqref{eq:intermed3} imply that 
\begin{align*}  \E_t[e^{iu_T  (\ov x_{t+T}-x_t)}] 
	&	 = \E_t[e^{iuX_T}A^{(2)}_T]+\E_t[e^{iuX_T}A^{(2)}_T (iuX^{(3)}_T-\tfrac12u^2Y_T)]+o^\mathrm{uc}(T^{2H})\\
	&= \E_t[e^{iuX_T}A^{(2)}_T]+\E_t[e^{iu_T\si_t(W_{t+T}-W_t)} (iuX^{(3)}_T-\tfrac12u^2Y_T)]+o^{\mathrm{uc}}(T^{2H}).\end{align*}
Notice that $\E_t[e^{iu_T\si_t(W_{t+T}-W_t)} (iuX^{(3)}_T-\tfrac12u^2Y_T)]=\wt C_2(u,T)_t$, so it remains to analyze the term $\E_t[e^{iuX_T}A^{(2)}_T]$. To this end, observe that $e^{iz+iw}-e^{iz}-iw=(e^{iz}-1)w+O(w^2)$ and therefore (recall \eqref{eq:mom4}, \eqref{eq:phipsi} and \eqref{eq:chi}),
\beq\label{eq:A20}\begin{split}
	A^{(2)}_T&= e^{T\psi_t(u_T)}\exp\biggl(\iint_t^{t+T} \Bigl(e^{iu_T\delta(t,z)+iu_T\ga(t,z)+iu_Tf'(v_t)\delta_v(t,z)(W_{t+T}-W_s)}\\
	&\quad-e^{iu_T(\delta(t,z)+\ga(t,z))}-iu_Tf'(v_t)\delta_v(t,z)(W_{t+T}-W_s)\Bigr)\nu(dz) ds\biggr)\\
	&=e^{T\psi_t(u_{T})} \exp\biggl(iu_Tf'(v_t)\chi^{(4)}_t(u_T)\int_t^{t+T} (W_{t+T}-W_s)ds+ O^{\mathrm{uc}}(T)\biggr).
\end{split}\eeq

Next, we note that by  \eqref{eq:mom3}, \eqref{eq:mom4} and the Cauchy--Schwarz inequality,
\beq\label{eq:chibound} \lvert \chi^{(4)}_t(u_T)\rvert \leq u_T\int_\R\lvert \delta(t,z)\delta_v(t,z)\rvert \nu(d z)+u_T^{q/2} \int_\R \lvert \delta_v(t,z)\rvert\lvert \ga(t,z)\rvert^{q/2}\nu(dz)=O^{\mathrm{uc}}(T^{-1/2}).\eeq
As a consequence, using the fact that $X_T=\si_tT^{-1/2}(W_{t+T}-W_t)+O(T^H)$ in the second step, we obtain
\begin{align*}
	&\E_t[e^{iuX_T}A^{(2)}_T]\\
	&\quad=e^{T\psi_t(u_{T})}\biggl(\E_t[e^{iuX_T}]+iu_Tf'(v_t)\chi^{(4)}_t(u_T)\E_t\biggl[e^{iuX_T}\int_t^{t+T} (W_{t+T}-W_s)ds\biggr]+ O^{\mathrm{uc}}(T)\biggr)\\
	&\quad=e^{T\psi_t(u_{T})}\biggl(\E_t[e^{iuX_T}]+iu_Tf'(v_t)\chi^{(4)}_t(u_T)\\
	&\qquad\times\E_t\biggl[e^{iu_T\si_t(W_{t+T}-W_t)}\int_t^{t+T} (W_{t+T}-W_s)ds\biggr] \biggr) +O^{\mathrm{uc}}(T^{1/2+H}).
\end{align*}
Now observe that $e^{T\psi_t(u_{T})}T\chi^{(4)}_t(u_T)=e^{T\varphi_t(u_{T})}T\chi^{(4)}_t(u_T)+O^\mathrm{uc}(T)$, which follows from \eqref{eq:chibound} and the bound
\beq\label{eq:phibound} \lvert \phi_t(u_T)\rvert \leq \frac{\lvert u\rvert^q}{T^{q/2}} \int_\R\lvert \ga(t,z)\rvert^q \nu(d z). \eeq
Moreover, since $\E[e^{iuX}X]=i^{-1}\frac{d}{du} \E[e^{iuX}] = -i^{-1}uve^{-\frac12 u^2v}$ for $X\sim N(0,v)$, we have
\begin{align*}
	&\E_t\biggl[e^{iu_T\si_t (W_{t+T}-W_t)}\int_t^{t+T} (W_{t+T}-W_s)ds\biggr]\\
	&\qquad=\int_t^{t+T} \E_t\Bigl[e^{iu_T\si_t (W_s-W_t)}\E_s[e^{iu_T\si_t (W_{t+T}-W_s)}(W_{t+T}-W_s)]\Bigr]ds=-\frac{u\si_tT^{3/2}}{2i} e^{-\frac12 u^2\si_t^2},
\end{align*}
which proves the expansion in the lemma. 

Finally, we have to justify the approximation  of $A^{(2)}_T=1+o^{\mathrm{uc}}(1)$ used above. By \eqref{eq:A20} and \eqref{eq:chibound}, we only need to show that 
\begin{align*} \E_t[ \lvert T\phi_t(u_T)\rvert]&=\E_t\biggl[ \biggl\lvert\int_\R T(e^{iu_T\ga(t,z)}-1)\nu(dz)\biggr\rvert\biggr]=O^{\mathrm{uc}}(T^{1/2}),\\
	\E_t[\lvert T\varphi_t(u_T)\rvert]&=\E_t\biggl[\biggl\lvert\int_\R T (e^{iu_T\delta(t,z)}-1-iu_T\delta(t,z))\nu(dz)\biggr\rvert\biggr]=o^{\mathrm{uc}}(1)\end{align*}
as $T\to0$. The first statement follows from \eqref{eq:mom4} and the bound
$$\E_t[\lvert T\phi_t(u_T)\rvert] \leq T^{1-q/2}u^q\E_t\biggl[\int_\R \lvert \ga(t,z)\rvert^q \nu(dz)\biggr]=O^{\mathrm{uc}}(T^{1/2}),$$
while the second statement follows from \eqref{eq:mom3} and the dominated convergence theorem.
\epr

To complete the proof of Theorem~\ref{thm:cf}, it remains to obtain explicit expressions for $\E_t[e^{iuX_T}]$ and $\wt C_2(u,T)_t$.
\blem\label{lem:expr} In the notation of \eqref{eq:X} and \eqref{eq:C0}, we have that 
\beq\label{eq:cfXT}\begin{split}
	\E_t[e^{iuX_T}]&=\exp\biggl(iu\al_t\sqrt{T}-\frac12u^2\int_0^1 \si^2_{t+sT\mid t} ds-iu^3\biggl(\frac12\si^2_tf'(v_t)\eta_t^v T^{1/2}+\frac{\si_t^2f'(v_t)\si^v_t}{\Ga(H+\frac{5}{2})}T^H\\
	&\quad+ C'_{1,0}(T)_tT^{2H} \biggr)\biggr) -e^{-\frac12\si_t^2 u^2} \biggl(\frac{u^2}{8H\Ga(H+\tfrac32)\Ga(H+\tfrac12)}-\frac{u^4\si_t^2}{\Ga(2H+3)}\\
	&\quad-\frac{u^4\si_t^2}{4(H+1)\Ga(H+\frac32)^2} \biggr)(f'(v_t)\si^v_t)^2T^{2H} +o^\mathrm{uc}(T^{2H})
\end{split}\raisetag{3\baselineskip}
\eeq
and 
\beq\label{eq:C0-2}  \begin{split}
	\wt C_2(u,T)_t&=e^{-\frac12u^2\si_t^2}\biggl( \frac{u^4\si_t^3f'(v_t)h'(\eta_t)\si^\eta_t}{\Ga(2H+3)}  \\
	&\quad+\frac{u^4\si_t^2[(f'(v_t)\wt \si^v_t)^2+\si_tf''(v_t)(\si^v_t)^2]}{4(H+1)\Ga(H+\frac32)^2}- \frac{u^2(f'(v_t)\wt\si^v_t)^2}{8H\Ga(H+\frac12)\Ga(H+\frac32)}  \biggr)T^{2H}.  \end{split}
\eeq
\elem
\bpr By \eqref{eq:X}, $X_T$ belongs to the direct sum of Wiener chaoses up to order $2$, where $X^{(0)}_T$, $X^{(1)}_T$ and $X^{(2)}_T$ belong to the zeroth, first and second Wiener chaos, respectively (see \cite[Proposition~1.1.4]{Nualart06}). Thus, we can use Theorem~\ref{thm:P2} to evaluate $\E_t[e^{iuX_T}]$. There are numbers $(\al_j)_{j\geq1}$ and $(\beta_j)_{j\geq1}$ (not related to those in the proof of Lemma~\ref{lem:xbar}, even though we use the same notation) such that 
\beq\label{eq:cfXT-2}\begin{split}\E_t[e^{iuX_T}]	&=\exp\biggl(iu\al_t\sqrt{T}-\frac12\sum_{j=1}^\infty \Bigl[\log(1-2i\al_ju)+2i\al_j u+\frac{\beta_j^2u^2}{1-2i\al_ju}\Bigr]\biggr)\\
	& =\exp\biggl(iu\al_t\sqrt{T}- u^2\sum_{j=1}^\infty \al_j^2 -\frac12u^2\sum_{j=1}^\infty \beta_j^2-iu^3\sum_{j=1}^\infty \al_j\beta_j^2  \\
	&\quad + 2u^4\sum_{j=1}^\infty \al_j^2\beta_j^2+4iu^5\sum_{j=1}^\infty \al_j^3\beta_j^2\biggr)+ O^\mathrm{uc}\biggl(\sum_{j=1}^\infty \al_j^3 \biggr). \end{split} \raisetag{-3.5\baselineskip}\eeq
The last term, that is, $4iu^5\sum_{j=1}^\infty \al_j^3\beta_j^2$ can be included in the $O^\mathrm{uc} (\sum_{j=1}^\infty \al_j^3  )$-term; we only spell it out because we referred to this in Section~\ref{sec:application}.
By \eqref{eq:square},
\begin{align*}
	\sum_{j=1}^\infty \al_j^2&=\frac12\E_t[(X^{(2)}_T)^2] 
	=\frac{(f'(v_t)\si^v_t)^2}{2\Gamma(H+\tfrac12)^2T}\int_t^{t+T}\int_t^s (s-r)^{2H-1} dr ds+O(T^{H+1/2}\vee T^{3H})\\
	&= \frac1{8H\Ga(H+\tfrac32)\Ga(H+\tfrac12)}(f'(v_t)\si^v_t)^2T^{2H}+O(T^{H+1/2}\vee T^{3H})
\end{align*}
and 
\begin{align*}
	\sum_{j=1}^\infty \beta_j^2=\E_t[(X_T^{(1)})^2] &= \frac{1}{T}\int_t^{t+T} \si_{s\mid t}^2 ds+o(T^{2H}) = \int_0^1 \si_{t+sT\mid t}^2 ds+o(T^{2H}). 
\end{align*}
In particular,
$
\sum_{j=1}^\infty \al_j^3 \leq  (\sum_{j=1}^\infty \al_j^2 )^{3/2}=O(T^{3H})$.
Moreover, by means of  \eqref{eq:square} and some tedious (but entirely straightforward) computations, one can show that
\begin{align*}
	\sum_{j=1}^\infty \al_j\beta_j^2&=\frac12T^{1/2}f'(v_t)\eta_t^v \Biggl(\int_0^1   \si_{t+sT\mid t}  ds\Biggr)^2 \\
	&\quad + \frac{f'(v_t)}{\Ga(H+\frac12)} T^H \int_0^1 \int_0^s(s-r)^{H-1/2}\si_{t+rT\mid t}\si^v_{t+rT\mid t} dr \si_{t+sT\mid t} ds\\
	&\quad+\frac{f''(v_t)\si^v_t}{\Ga(H+\frac12)} T^H\int_0^1\int_0^s (s-r)^{H-1/2}\si_{t+rT\mid t}\si_{t+sT\mid t} (v_{t+sT\mid t}-v_t) drds  +o(T^{2H}) \\
	&=\frac12\si^2_tf'(v_t)\eta_t^vT^{1/2}  +\frac{\si_t^2f'(v_t)\si^v_t}{\Ga(H+\frac{5}{2})}T^H + C'_{1,0}(T)_t T^{2H} + o(T^{2H}),\\
	\sum_{j=1}^\infty \al_j^2\beta_j^2	&=\frac{1}{2T^2}\int_t^{t+T} \int_t^s \si_{r\mid t}\si_{s\mid t}\int_t^s \Bigl(f'(v_t)\eta^v_t+k_H(s-w)[f'(v_t)\si^v_{w\mid t}+f''(v_t)\si^v_t (v_{s\mid t}-v_t)]\Bigr)\\
	&\qquad\quad\times\Bigl(2f'(v_t)\eta^v_t+k_H(r-w)[f'(v_t)\si^v_{w\mid t}+f''(v_t)\si^v_t(v_{r\mid t}-v_t)]\\
	&\qquad\qquad\qquad+ k_H(w-r)[f'(v_t)\si^v_{r\mid t}+f''(v_t)\si^v_t(v_{w\mid t}-v_t)]\Bigr) dw drds+o(T^{2H})\\
	&= \frac{\si_t^2(f'(v_t)\si^v_t)^2}{2T^2}\int_t^{t+T} \int_t^s  \int_t^s k_H(s-w)[k_H(r-w) dw k_H(w-r)] dwdrds+ o(T^{2H})\\
	&=\biggl(\frac{1}{2\Ga(2H+3)} + \frac{1}{8(H+1)\Ga(H+\frac32)^2}\biggr) \si_t^2(f'(v_t)\si^v_t)^2T^{2H}+ o(T^{2H}).
\end{align*}
Consequently,
\begin{align*}
	\E_t[e^{iuX_T}]&=\exp\biggl(iu\al_t\sqrt{T}-\frac12u^2\int_0^1 \si^2_{t+sT\mid t} ds-iu^3\biggl(\frac12\si^2_tf'(v_t)\eta_t^v T^{1/2}+\frac{\si_t^2f'(v_t)\si^v_t}{\Ga(H+\frac{5}{2})}T^H\\
	&\quad+ C'_{1,0}(T)_tT^{2H} \biggr)\biggr) \biggl(1-\frac{u^2(f'(v_t)\si^v_t)^2T^{2H}}{8H\Ga(H+\tfrac32)\Ga(H+\tfrac12)}+\biggl(\frac{1}{\Ga(2H+3)} \\
	&\quad+ \frac{1}{4(H+1)\Ga(H+\frac32)^2}\biggr)u^4\si_t^2(f'(v_t)\si^v_t)^2T^{2H}  \biggr) +o^{\mathrm{uc}}(T^{2H})\\
	&=\exp\biggl(iu\al_t\sqrt{T}-\frac12u^2\int_0^1 \si^2_{t+sT\mid t} ds-iu^3\biggl(\frac12\si^2_tf'(v_t)\eta_t^v T^{1/2}+\frac{\si_t^2f'(v_t)\si^v_t}{\Ga(H+\frac{5}{2})}T^H\\
	&\quad+ C'_{1,0}(T)_tT^{2H} \biggr)\biggr) -e^{-\frac12 \si_t^2 u^2}\biggl(\frac{u^2}{8H\Ga(H+\tfrac32)\Ga(H+\tfrac12)}-\biggl(\frac{1}{\Ga(2H+3)} \\
	&\quad+ \frac{1}{4(H+1)\Ga(H+\frac32)^2}\biggr) u^4\si_t^2\biggr) (f'(v_t)\si^v_t)^2T^{2H}+o^{\mathrm{uc}}(T^{2H}),
\end{align*}
which proves \eqref{eq:cfXT}. 

In order to find $\wt C_2(u,T)_t$, we first determine $\E_t[e^{iu_T\si_t (W_{t+T}-W_t)}X^{(3)}_T]$. To this end, let $M_\tau=e^{iu_T\si_t(W_\tau-W_t)-\frac12(iu_T\si_t)^2 (\tau-t)}$, which is an exponential martingale and solves the SDE
$$ dM_\tau=iu_T\si_t M_\tau dW_\tau,\quad\tau\geq t,\quad M_t=1. $$
In integral form, this becomes
$$ M_\tau= 1+ iu_T\si_t \int_t^\tau M_s dW_s.$$
So upon iterating this equation, we obtain the chaos expansion of $M$ as
$$ M_\tau=1+\sum_{n=1}^\infty (iu_T\si_t)^n \int_t^\tau \int_t^{t_1} \dotsi \int_t^{t_{n-1}} dW_{t_n}\cdots dW_{t_2}dW_{t_1}.$$
Now observe that $\E_t[e^{iu_T\si_t (W_{t+T}-W_t)}X^{(3)}_T]=e^{-\frac12 u^2\si_t^2}\E_t[M_{t+T}X^{(3)}_T]$. Because Wiener chaoses are orthogonal to each other (see \eqref{eq:ortho}), to compute the last expectation, only the projection of $M_{t+T}$ on the third-order Wiener chaos has a non-zero contribution. Therefore,
\begin{align*}
	&\E_t[M_{t+T}X^{(3)}_1]\\
	&\quad=\frac{(iu_T\si_t)^3}{T^{1/2}} \E_t\Biggl[\int_t^{t+T} \int_t^s   \int_t^r dW_w dW_r dW_s\int_t^{t+T}\int_t^s \int_t^r\Bigl(\la_t\si^\eta_tk_H(s-r)k_H(r-w) \\
	&\qquad\qquad\qquad\qquad\qquad\qquad\qquad\qquad  +f''(v_t)(\si^v_t)^2k_H(s-r)k_H(s-w)\Bigr) d  W_w dW_r dW_s\Biggr]\\
	&\quad=\frac{(iu_T\si_t)^3}{T^{1/2}}   \int_t^{t+T}\int_t^s \int_t^r k_H(s-r)[\la_t\si^\eta_tk_H(r-w)+f''(v_t)(\si^v_t)^2k_H(s-w)] dwdrds\\
	& \quad= -\frac{iu^3\la_t\si_t^3\si_t^\eta}{\Ga(2H+3)}T^{2H}-\frac{iu^3 f''(v_t)\si_t^3(\si_t^v)^2}{4(H+1)\Ga(H+\frac32)^2}T^{2H}.
\end{align*}
In a similar fashion, we can find $\E_t[e^{iu_T\si_t (W_{t+T}-W_t)}Y_T]$. First, note that $Y_T=Y_T^{(0)}+Y_T^{(2)}$, where 
\begin{align*}
	Y_T^{(0)} &=T^{-1}(f'(v_t)\wt\si^v_t)^2 \int_t^{t+T}\int_w^{t+T} k_H^2(s-w) ds dw = \frac{(f'(v_t)\wt\si^v_t)^2}{4H\Ga(H+\frac12)\Ga(H+\frac32)}T^{2H},\\
	Y_T^{(2)}&= 2T^{-1}(f'(v_t)\wt\si^v_t)^2 \int_t^{t+T}\int_w^{t+T}\int_w^s k_H(s-w)k_H(r-w)dW_rdW_sdw\\
	&=\frac{2T^{-1}(f'(v_t)\wt\si^v_t)^2}{\Ga(H+\frac12)^2}\int_t^{t+T}\int_t^s \int_t^r (s-w)^{H-1/2}(r-w)^{H-1/2} dw dW_rdW_s.
\end{align*}
Thus, 
\begin{align*}
	&\E_t[e^{iu_T\si_t (W_{t+T}-W_t)}Y_T]\\
	&\quad=(f'(v_t)\wt\si^v_t)^2 e^{-\frac12u^2\si_t^2} \biggl( \frac{T^{2H}}{4H\Ga(H+\frac12)\Ga(H+\frac32)}\\
	&\qquad\qquad\qquad\qquad -\frac{2T^{-2}u^2\si_t^2}{\Ga(H+\frac12)^2}\int_t^{t+T}\int_t^s \int_t^r (s-w)^{H-1/2}(r-w)^{H-1/2} dw drds\biggr)\\
	&\quad=(f'(v_t)\wt\si^v_t)^2 e^{-\frac12u^2\si_t^2} \biggl( \frac{T^{2H}}{4H\Ga(H+\frac12)\Ga(H+\frac32)} -\frac{u^2\si_t^2T^{2H}}{2(H+1)\Ga(H+\frac32)^2} \biggr).
\end{align*}
Gathering the formulas derived above, we obtain \eqref{eq:C0-2}.
\epr

\bpr[Proof of Theorem~\ref{thm:cf}]
By Lemmas~\ref{lem:xcheck}--\ref{lem:expr}, we have that
\begin{align*}
	&\E_t[e^{iu_T(x_{t+T}-x_t)}]\\
	&\quad=e^{T\psi_t(u_{T})} \biggl\{\exp\biggl(iu\al_t\sqrt{T}-\frac12u^2\int_0^1 \si^2_{t+sT\mid t} ds-iu^3\biggl(\frac12\si^2_tf'(v_t)\eta_t^v T^{1/2}+\frac{\si_t^2f'(v_t)\si^v_t}{\Ga(H+\frac{5}{2})}T^H\\
	&\quad\qquad+ C'_{1,0}(T)_tT^{2H} \biggr)\biggr) -e^{-\frac12\si_t^2 u^2} \biggl(\frac{u^2}{8H\Ga(H+\tfrac32)\Ga(H+\tfrac12)}- \frac{u^4\si_t^2}{\Ga(2H+3)}\\
	&\quad\qquad-\frac{u^4\si_t^2}{4(H+1)\Ga(H+\frac32)^2} \biggr)(f'(v_t)\si^v_t)^2T^{2H} \biggr\}\\
	&\quad\quad	+ e^{-\frac12\si_t^2 u^2+T\varphi_t(u_{T})}\biggl(\sum_{i=1}^4 C_{1,i}(u,T)_t + C'_{1,1}(u,T)_t\biggr)+ \wt C_2(u,T)_t+o^{\mathrm{uc}}(T^{2H})\\
	&\quad=\exp\biggl(iu\al_t\sqrt{T}-\frac12 u^2\int_0^1 \si^2_{t+sT\mid t} ds+T\psi_t( {u}_{T}) -iu^3\biggl(\frac12\si^2_tf'(v_t)\eta_t^v T^{1/2}+\frac{\si_t^2f'(v_t)\si^v_t}{\Ga(H+\frac{5}{2})}T^H\\
	&\quad\qquad+ C'_{1,0}(T)_tT^{2H} \biggr) \biggr) + e^{-\frac12\si_t^2 u^2+T\varphi_t(u_{T})}\biggl(\sum_{i=1}^4 C_{1,i}(u,T)_t+   C'_{1,1}(u,T)_t \biggr)\\
	&\qquad+C_2(u)_tT^{2H}+o^{\mathrm{uc}}(T^{2H}),
\end{align*}
which shows Theorem~\ref{thm:cf}.
\epr

\begin{proof}[Proof of Proposition~\ref{prop:condmean}]
	By \eqref{eq:x}, we have that 
	\begin{align*} \E_t[x_{t+T}-x_t] &= \biggl(\al_t+\int_\R\ga(t,z)\nu(dz)\biggr)T \\
		&\quad+ \E_t\biggl[ \int_t^{t+T} (\al_s-\al_t) ds + \iint_t^{t+T} (\ga(s,z)-\ga(t,z)) ds\nu(dz) \biggr]. \end{align*}
	Both $\al$ and $\ga$ are $H$-H\"older continuous in $L^1$ by \eqref{eq:smooth2} and  \eqref{eq:smooth3} (and our hypotheses on $q$ and $H_\ga$). Therefore, the conditional expectation on the right-hand side above is $O(T^{1+H})$, which gives the desired result.
\end{proof}

\section{Proof of Theorem~\ref{thm:err_bound}}\label{sec:proof:err_bound}
Throughout the proof, we will set $t=0$ and $x_0=0$ for simplicity of notation. We will also drop $t$ from the notation. We will make use of the following lemma in the proof:
\begin{lemma}\label{lemma:bounds}
	Suppose that Assumptions \ref{ass:A}, \ref{ass:B} and \ref{ass:C}-1 hold. There exists $\mathcal{F}_0^{(0)}$-measurable random variables $C_0$ and $\overline{t}>0$ that do not depend on $T$ such that for $T<\overline{t}$, we have 
	\begin{equation}\label{bounds_1}
		O_T(k)\leq C_0\left(Te^{3k}1_{\{k<-1\}}+Te^{-k}1_{\{k>1\}}+\left(\sqrt{T}\wedge\frac{T}{|k|}\right)1_{\{|k|<1\}}\right),
	\end{equation}
	\begin{equation}\label{bounds_2}
		|O_{T}(k_1)-O_T(k_2)|\leq C_0\left[\frac{T}{k_2^4}\wedge\frac{T}{k_2^2}\wedge 1\right]|e^{k_1}-e^{k_2}|,
	\end{equation}
	where $k_1<k_2<0$ or $k_1>k_2>0$. 
\end{lemma}

\bpr[Proof of Lemma~\ref{lemma:bounds}] By an  application of It\^o's lemma, we have 
\begin{align*}
	e^{x_t}-1 &= \int_0^te^{x_s}\alpha_sds + \int_0^te^{x_s}\sigma_sdW_s 
	+\int_0^t\int_{\mathbb{R}}e^{x_{s-}}(e^{\gamma(s,z)}-1)\mu(ds,dz)\\&~ + \int_0^t\int_{\mathbb{R}}e^{x_{s-}}(e^{\delta(s,z)}-1 - \delta(s,z))(\mu-\nu)(ds,dz) + \frac{1}{2}\int_0^te^{x_s}\sigma_s^2ds,
\end{align*}
and we have an analogous expression for $e^{-x_t}-1$. By the Cauchy--Schwarz inequality and the integrability conditions in Assumptions \ref{ass:B} and \ref{ass:C}-1, we then have 
\beq\label{eq:bounds_1}
\mathbb{E}_0[(e^{x_s}-1)^2] + \mathbb{E}_0[(e^{-x_s}-1)^2]\leq C_0 T,\quad s\in[0,T\wedge \overline{t}].
\eeq
To proceed further, we make use of the following algebraic inequalities:
\beq\label{eq:bounds_2}
(e^k-e^x)^+\leq 2e^{2k}\frac{|e^{-x}-1|^2}{|e^{-k}-1|},~k<0,~x\in\mathbb{R},
\eeq
\beq\label{eq:bounds_3}
(e^x-e^k)^+\leq 2\frac{|e^{x}-1|^2}{|e^{k}-1|},~k>0,~x\in\mathbb{R},
\eeq
with the notation $x^+ = \max(x,0)$. Using this result and (\ref{eq:bounds_1}), we get the bound in (\ref{bounds_1}). We proceed with showing (\ref{bounds_2}). For this, we make use of the following inequalities:
\beq\label{eq:bounds_4}
|(X-K_1)^+-(X-K_2)^+|\leq |K_1-K_2|\bone_{\{X>K_2\}},\quad X\in\mathbb{R},\quad K_1\geq K_2,
\eeq
\beq\label{eq:bounds_5}
|(K_1-X)^+-(K_2-X)^+|\leq |K_1-K_2|\bone_{\{X<K_2\}},\quad X\in\mathbb{R},\quad K_1\leq K_2.
\eeq
From here, to show (\ref{bounds_2}), we need to bound $\mathbb{Q}_0(|x_T|>k)$ for any $k>0$. An application of (\ref{eq:bounds_1}) yields
\beq\label{eq:bounds_6}
\mathbb{E}_0 [|x_s|^4]\leq C_0 T,\quad s\in[0,T\wedge \overline{t}].
\eeq
Combining (\ref{eq:bounds_6}) with the inequalities in (\ref{eq:bounds_4})--(\ref{eq:bounds_5}), we get (\ref{bounds_2}).\epr

We are now ready to prove Theorem~\ref{thm:err_bound}. We can make the following decomposition:
\begin{equation}
	\widehat{\mathcal{L}}_{T}(u) - \mathcal{L}_{T}(u) = \zeta_{T}^{(1)}(u)+\zeta_{T}^{(2)}(u)+\zeta_{T}^{(3)}(u),
\end{equation}
where
\begin{align*}
	\zeta_{T}^{(1)}(u) &= \left(\frac{u^2}{T}+i\frac{u}{\sqrt{T}}\right)\left(\int_{-\infty}^{\underline{k}}e^{(iu/\sqrt{T}-1)k}O_{T}(k)dk+\int_{\overline{k}}^{\infty}e^{(iu/\sqrt{T}-1)k}O_{T}(k)dk\right),\\
	\zeta_{T}^{(2)}(u) &= -\left(\frac{u^2}{T}+i\frac{u}{\sqrt{T}}\right)\sum_{j=2}^N\int_{k_{j-1}}^{k_j}\left(e^{(iu/\sqrt{T}-1)k_{j-1}}O_T(k_{j-1})-e^{(iu/\sqrt{T}-1)k}O_T(k)\right)dk,\\
	\zeta_{T}^{(3)}(u) &= -\left(\frac{u^2}{T}+i\frac{u}{\sqrt{T}}\right)\sum_{j=2}^Ne^{(iu/\sqrt{T}-1)k_{j-1}}\epsilon_{T}(k_{j-1})\Delta_j.
\end{align*}
Using (\ref{bounds_1}), we have 
\begin{equation}
	\zeta_{T}^{(1)}(u) = O_{\ov \P}^\mathrm{uc}\left(e^{2\underline{k}}+e^{-2\overline{k}}\right).
\end{equation}
Next, using (\ref{bounds_1}) and (\ref{bounds_2}), we get 
\begin{equation}
	\zeta_{T}^{(2)}(u) = O_{\ov \P}^\mathrm{uc}\left(\frac{\Delta}{\sqrt{T}}\log T\right).
\end{equation}
Finally, using our assumption for $\mathcal{F}$-conditional independence of the option observation errors and the bounds in (\ref{bounds_1}) and (\ref{bounds_2}), we have 
\begin{equation}
	\begin{split}
		&\mathbb{E}^{\ov\P}_0[|\zeta_{T}^{(3)}(u)|^2]\leq C_0(|u|^3\vee 1)\frac{\Delta}{\sqrt{T}},\\&\mathbb{E}^{\ov\P}_0[|\zeta_{T}^{(2)}(u)-\zeta_{T}^{(2)}(v)|^4]\leq C_0(|u|\vee |v|\vee 1)^8(|u-v|^4\vee|u-v|^8)\frac{\Delta^2}{T},
	\end{split}
\end{equation} 
for some $\mathcal{F}_0$-adapted random variable $C_0$ (that does not depend on $u$ and $v$). From here, we have the tightness of $\frac{T^{1/4}}{\sqrt{\Delta}}\zeta_{T}^{(3)}(u)$ in the space of continuous functions of $u$ equipped with the local uniform topology. Combining the above three bounds, we have the result of the theorem about $\widehat{\mathcal{L}}_{T}(u) - \mathcal{L}_{T}(u)$. The result for $\widehat{\mathcal{M}}_{T} - \mathcal{M}_{T}$ in exactly the same way.

\section{Proofs for Section~\ref{sec:application}}\label{sec:9}

\begin{proof}[Proof of Corollary~\ref{cor:arg}]
	Since there are no jumps, $\psi_t$ as well as the terms $C_{1,j}(u,T)_t$, for $j=1,\dots,4$, and $C'_{1,1}(u,T)_t$ are identically zero. Therefore, by \eqref{eq:cf},
	\[ \Arg(\call_{t,T}(u))= u\al_t T^{1/2} - \frac{\si_t^2f'(v_t)\si^v_tu^3}{\Ga(H+\frac52)} T^H -\frac12u^3\si_t^2f'(v_t)\eta_t^v T^{1/2}+O^\mathrm{uc}(T^{2H}).   \]
	By Proposition~\ref{prop:condmean}, we thus have
	\[ A_{t,T}(u) =  - \frac{\si_t^2f'(v_t)\si^v_tu^3}{\Ga(H+\frac52)} T^H -\frac12u^3\si_t^2f'(v_t)\eta_t^v T^{1/2}+O^\mathrm{uc}(T^{2H}),\]
	from which all statements of the corollary follow.
\end{proof}

\begin{proof}[Proof of Theorem~\ref{thm:Hest}]
	By Theorem~\ref{thm:err_bound} and the mean-value theorem, we have
	\[ \wh A_{t,T}(\un u) = A_{t,T}(\un u) + O(\Delta^{1/2}T^{-1/4-2H}). \]
	Thus, by another application of the mean-value theorem, we deduce from 
	\eqref{eq:ratio} that
	\begin{align*}	\wh H_n &= \frac{\log  A_{t,T_1}(\underline u) - \log  A_{t,T_2}(\underline u)}{\log \tau} + O_{\ov\P}(\Delta^{\frac12}T_1^{-\frac14-H}) \\
		&= H + O_{\ov\P}(T^{(\frac12-H)\wedge H}\vee  \Delta^{\frac12}T_1^{-\frac14-H}) , \end{align*}
	which proves \eqref{eq:H-conv}. Equation \eqref{eq:H-conv-2} is proved analogously.
\end{proof}

\begin{proof}[Proof of Theorem~\ref{thm:secdiff}]
	The expansion \eqref{eq:threeT} follows from \eqref{eq:cf} and the fact that the second-order differences are constructed in such a way that the drift $\al_t \sqrt{T}$, the jump component $T\psi_t(u/\sqrt{T})$, the leverage component $-\frac12iu^3\si_t^2f'(v_t)\eta_t^v T^{1/2}$ and the terms $C_{1,j}(u,T)_t$ for $j=2,3,4$ are canceled out perfectly. What remains in the argument are therefore terms of order $T^{2H}$ and $C_{1,1}(u,T)_t$ and $C'_{1,1}(u,T)_t$, both of which are $O^\mathrm{uc}(T^{H_\delta+1-r/2})$. 
	The second claim is an easy consequence of \eqref{eq:threeT} and the mean-value theorem.
\end{proof}

\begin{appendix}
	\section*{Appendix: Some elements of Wiener space theory}\label{appn}
	\setcounter{theorem}{0}
	\setcounter{equation}{0}
	
	Let $(W_t)_{t\in[0,1]}$ be a one-dimensional standard Brownian motion 
	and consider the Gaussian Hilbert space $\calh=\{\int_0^1 f(t) dW_t : f\in L^2([0,1])\}\subset L^2(\Om)$. Let $H_0(x)=1$ and $H_n(x)=\frac{(-1)^n}{n!}e^{x^2/2}\frac{d^n}{dx^n} e^{-x^2/2}$  be the \emph{$n$th Hermite polynomial} for $n\geq1$. The first four Hermite polynomials are
	\beq\label{eq:Hn}  H_0(x)=1,\quad H_1(x)=x, \quad H_2(x)=\frac12(x^2-1), \quad H_3(x)=\frac16x^3-\frac12x. \eeq
	For $n\geq0$, define the \emph{$n$th Wiener chaos
		$\calh_n$} as the $L^2$-closure of the linear subspace $\{H_n(\int_0^1 f(t)dW_t):  h\in L^2([0,1]),\int_0^1h^2(t)dt =1\}$ and further let $\calp_n = \bigoplus_{i=0}^n \calh_i$. It is well-known that Wiener chaoses of different order are orthogonal to each other:
	\beq\label{eq:ortho} X\in\calh_n,\ Y\in \calh_m,\ n\neq m \implies \E[XY]=0;\eeq
	cf.\ \cite[Theorem~1.1.1]{Nualart06}. While there is no general characterization of the distribution of elements in $\calh_n$, the case $n=2$ is special. In fact, the distribution of elements in $\calp_2$ is explicitly known:
	\bthm\label{thm:P2} If $X\in\calp_2$, then are real numbers $(\al_j)_{j\geq1}$ and $(\beta_j)_{j\geq1}$ satisfying $\sum_{j=1}^\infty \al_j^2<\infty$ and $\sum_{j=1}^\infty \beta_j^2<\infty$ such that
	\beq\label{eq:cfX} \E[e^{iuX}]=\exp\biggl(i\E[X]u-\frac12\sum_{j=1}^\infty \Bigl[\log(1-2i\al_ju)+2i\al_j u+\frac{\beta_j^2u^2}{1-2i\al_ju}\Bigr]\biggr). \eeq
	Moreover, there exists a sequence of independent standard normal random variables $(\xi_n)_{n\geq1}$ such that if we write $X=X_0+X_1+X_2$ with $X_0\in\calh_0$, $X_1\in\calh_1$ and $X_2\in\calh_2$, then $X_0=\E[X]$, $X_1=\sum_{j=1}^\infty \beta_j\xi_j$ and $X_2=\sum_{j=1}^\infty \al_j(\xi^2_j-1)$.
	In particular, we have that 
	\beq\label{eq:square} \begin{split}
		&	\E[X_1^2]=\sum_{j=1}^\infty \beta_j^2, \quad\E[X_2^2]=2\sum_{j=1}^\infty\al_j^2,\quad	\E[X_2^3]=8\sum_{j=1}^\infty \al_j^3,\quad\E[X_1^2X_2]=2\sum_{j=1}^\infty \al_j\beta_j^2, \\
		&	 \mathrm{Cov}(X_1^2,X_2^2)=8\sum_{j=1}^\infty \al_j^2\beta_j^2,\quad \mathrm{Cov}(X_1^2,X_2^3)-3\E[X_1^2X_2]\E[X_2^2]=48\sum_{j=1}^\infty \al_j^3\beta_j^2.
	\end{split}\!\!\!\!\eeq
	\ethm
	\begin{proof}
		The first two claims and the first two formulas in \eqref{eq:square} are shown in   \cite[Theorem~6.1]{Janson97}. The last two formulas in \eqref{eq:square} follow from a straightforward calculation based on the series representation of $X_1$ and $X_2$.
	\end{proof}

\end{appendix}

\begin{acks}[Acknowledgments]
	The authors would like to thank two anonymous referees, an Associate
	Editor and the Editor for their constructive comments that improved the
	quality of this paper.
\end{acks}

\bibliographystyle{imsart-nameyear} 
\bibliography{rough}
 
\end{document}